\documentclass{article}
\usepackage{amsmath,amssymb,amsthm,cases}
\usepackage[utf8]{inputenc}
\usepackage[margin=1in]{geometry}
\usepackage{algorithm}

\usepackage{microtype}
\usepackage{graphicx}
\usepackage{subfigure}
\usepackage{booktabs}
\usepackage{natbib}
\setcitestyle{open={(},close={)}}

\usepackage[colorlinks=true,citecolor=blue]{hyperref}

\newtheorem{theo}{Theorem}

\newtheorem{prop}{Proposition}

\newcommand*{\defeq}{\stackrel{\text{def}}{=}}

\begin{document}

\title{\textbf{Multi-view Geometry: Correspondences Refinement Based on Algebraic Properties}}

\author{\vspace{0.5in}\\\textbf{Trung-Kien Le} ~and~ \textbf{Ping Li}\vspace{0.2in}\\
Cognitive Computing Lab\\
 Baidu Research\\
 10900 NE 8th St. Bellevue, WA 98004, USA\\\\
 \{hieukien1207,\ pingli98\}@gmail.com\\\\
}

\date{}

\maketitle

\begin{abstract}
\noindent Correspondences estimation or feature matching is a key step in the image-based 3D reconstruction problem. In this paper, we propose two algebraic properties for correspondences. The first is a rank deficient matrix construct from the correspondences of at least nine key-points on two images (two-view correspondences) and the second is also  another rank deficient matrix built from the other correspondences of six key-points on at least five images (multi-view correspondences). To our knowledge, there are no theoretical results for multi-view correspondences prior to this paper. To obtain accurate correspondences, multi-view correspondences seem to be more useful than two-view correspondences. From these two algebraic properties, we propose an refinement algorithm for correspondences. This algorithm is a combination of correspondences refinement, outliers recognition and missing key-points recovery. Real experiments from the project of reconstructing Buddha statue show that the proposed refinement algorithm can reduce the average error from 77 pixels to 55 pixels  on the correspondences estimation. This drop is substantial and it validates our results.
\end{abstract}


\newpage

\section{Introduction}
Image-based 3D reconstruction is a classical and longstanding ill-pose problem in the computer graphics and computer vision. It easily finds lots of 3D reconstruction applications on robotics and autonomous devices via the \emph{simultaneous localization and mapping} (SLAM) topic~\citep{Wang2007,Cadena2016,Bresson2017}, on medical imaging via the \emph{human organ modeling} problem~\citep{McInerney1996,Heimann2009,Francisco2014}, on physical geography~\citep{Smith2016}, and many others. In addition to wide applications in science, the image-based 3D reconstruction has also received great attention from scientists because of the fun it brings. Imagine how happy computer vision researchers with a passion for art would be when they recreate stunning statues like Venus de Milo and Discobolus and store them as 3D point clouds. Or how happy computer graphics researchers with a passion for architecture would be when they can build some beautiful buildings like the Taj Mahal, the Colosseum, the Forbidden City and the Blue Mosque as 3D point clouds models and everyday enjoy them with a huge screen. With the rapid development of the computer graphics and vision, we believe that the happiness of these researchers will soon become a reality.

The image-based 3D reconstruction of an object can be seen as a complete process starting from multiple 2D images and ending with modeled 3D point clouds of this object. In general, this complete process is a combination of the five stages that (i) the key-points or the features detection~\citep{Lowe1999,Lowe2004,Yan2004,Mikolajczyk2004,Bay2006}, (ii) the point correspondences estimation, the feature matching or the point set registration~\citep{Munkres1957,Myronenko2010,Torresani2012,Yan2016,Swoboda2019}, (iii) the camera calibration or the projection matrix estimation~\citep{Weng1992,Wei1994,Zhang2000,Remondino2006,Chuang2021}, (iv) the triangulation~\citep{Hartley1997triangulation,Stewenius2005,yang2019,Sharp2020}, and (v) the bundle adjustment~\citep{Triggs1999,Agarwal2010,Zach2014}. These five stages are mutually dependent. Depending on the reconstructed object, the features detecting by the first stage can be points~\citep{Hsieh1996}, edges~\citep{Ziou1998,dollar2013}, corners~\citep{Zheng1999,Rosten2006}, circles~\citep{Alhichri2003}, or special curves. Methods for matching features are depended on the numbers and the characteristics of the features. For instance, when the numbers of features are in the tens or hundreds, the methods in the `graph matching' group~\citep{Munkres1957,Zhou2013,Swoboda2019} are prefer. However, when these numbers are in the thousands or more, the methods in the `point set registration' group~\citep{Jian2010,Myronenko2010,Zhou2019} seem better. If the features are too many, the evaluating `point-match-point' should be naturally replaced with `pointsss-match-pointsss' where `pointsss' means a lot of points. In view of this, some probabilistic models such as Gaussian mixture models and Bayesian mixture models are used to characterize these evaluated pointsss. Of the five stages of image-based 3D reconstruction process, the correspondences estimation or the feature matching is the most difficult and has been received the most attention from researchers. For example, the well-known formulation of the graph matching problem is the \emph{quadratic assignment problem} (QAP) known as NP-hard~\citep{Lawler1963}. In addition, the feature matching across multiple images is still a challenging task although several potential methods exist for this on two images.

\begin{figure*}[t]
  \centering
  \includegraphics[width=0.8\linewidth]{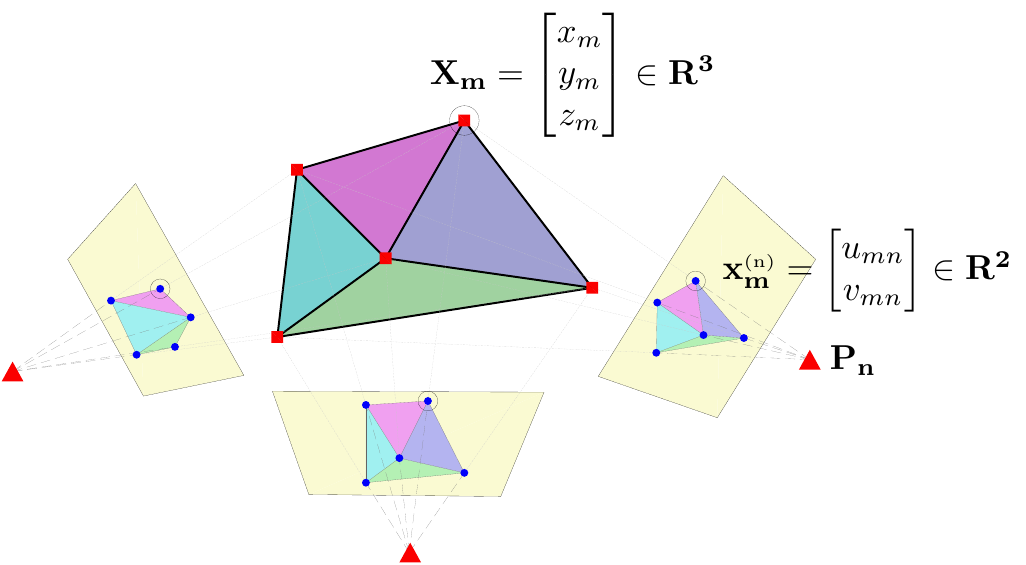}
  \caption{Multi-view geometry: `\emph{Red squares}' are the world points, `\emph{red triangles}' are the centers of projection (COPs). They are unknown and from that \emph{camera calibration} means estimating the positions of the red triangles and \emph{3D reconstruction} means estimating the red squares' positions. `\emph{Light yellow planes}' are the image and `\emph{blue circles}' are the image points. They are known. The three blue circles marked by `\emph{black circles}' are three-view correspondence of the world point ${\bf X}_m$.}
  \label{Fig:Multi_view_geo}
\end{figure*}

Using the correspondences from the first and second stages, the fundamental matrices of each pair images are computed by the \emph{eight-point algorithm}~\citep{Hartley2004}, the \emph{five-point algorithms}~\citep{Nister2004,Barath2018} or others. Fixing the projection matrix of the reference image, other projection matrices of other images are determined by these fundamental matrices. From the projection matrices, positions of camera centers are obtained. This is a reason the third stage named by the `camera calibration'. The final 3D point clouds are calculated from the projection matrices and refined by the triangulation and bundle adjustment stages.

Set in a vast forest with a constant fear of getting lost, this paper attempts to study 3D reconstruction from the basic and theoretical point. First, we mention that if the correspondences from the second stage are very good, some of the existing methods in the camera calibration, the triangulation and the bundle adjustment can give us precise 3D point clouds. The interesting thing is that we don't need too many, just in dozens of good correspondences. Second, we have a question about getting dozens of good correspondences from images. Of course, good features or good key-points bring us good correspondences. A `good key-point' means it is very different from other points. In other words, a key-point should have enough characteristics to be recognized. In view of this, we prefer to study the image-based 3D reconstruction on multiple images rather than just two images. From that the characteristics of the key-points are built based on multiple images, not only two images. More precisely, based on the \emph{multi-view geometry} theory for the image-based 3D reconstruction~\citep{Hartley2004}, we study the algebraic properties of the correspondences. Two rank properties, one for the correspondences of at least nine key-points on two images and another for the correspondences of six key-points on at least five images, are derived. These two rank properties give us two constraints among correspondences. They are help us not only on refining correspondences, but also on recognizing outliers and recovering missing key-points.

The applications of our proposed results, i.e., refining correspondences, recognizing outliers and recovering missing key-points are validated by some realistic experiments from our small project on 3D reconstructing the Buddha statue. In these experiments, we manually determine the ground-truths for the correspondences. Therefore, all our outcome evaluations based on the differences between the ground-truths and the estimations, the refinements are reliable. From twenty one experiments, the average error between the ground-truths and the estimated correspondences are 77 pixels. After applying our proposed refinement algorithm, this average error is 55 pixels. This drop is really dramatic and it validates our results.\\

\textbf{Roadmap.}\ In Section~\ref{sec:multi-view_geometry}, we introduce the multi-view geometry problem, the $N$-view and $M$-point camera calibration, the importance of accurate correspondences. Then Section~\ref{sec:rank_properties} derives the two rank properties for correspondences,  presented by Theorem~\ref{Thm:RanktwoView} and Theorem~\ref{Thm:RankfiveView}, which are the main results of the paper. In Section~\ref{sec:refinement}, we apply these two rank properties and present the main correspondences refinement algorithm (Algorithm~\ref{Alg:Main}). Section~\ref{sec:evaluation} provides and evaluates results, and finally Section~\ref{sec:conclusion} concludes this article.


\section{Multi-view Geometry}\label{sec:multi-view_geometry}

The \emph{multi-view geometry} is illustrated by Figure~\ref{Fig:Multi_view_geo}. There are $M$ unknown world points ${\bf X}_1, {\bf X}_2, \ldots, {\bf X}_M$ in the three-dimensional space $\mathbb{R}^3$. The three real numbers $x_m, y_m, z_m$ are used to indicate the three coordinates of the world point ${\bf X}_m$. These $M$ world points are captured by $N$ images. The $n$th image capturing the world points  means that there is a $3\times 4$ projection matrix ${\bf P}_n$, formulated by
\begin{equation}\label{eq01}
{\bf P}_n \defeq \begin{bmatrix}p_{1n} & p_{2n} & p_{3n} & p_{4n}\\
p_{5n} & p_{6n} & p_{7n} & p_{8n}\\
p_{9n} & p_{10n} & p_{11n} & p_{12n}\end{bmatrix}_{3\times 4\, ,}
\end{equation}
to map all the world points to a plane which is called the $n$th \emph{image-plane}. Using the homogeneous coordinate system, the mapping from the world point ${\bf X}_m = [x_m,y_m,z_m]^T$ to the image point ${\bf x}^{\text{\tiny (n)}}_m \defeq [u_{mn},v_{mn}]^T$ by the projection matrix ${\bf P}_n$ is presented mathematically as follows
\begin{equation}\label{eq02}
{\bf P}_n\begin{bmatrix}{\bf X}_m\\ 1\end{bmatrix} = \begin{bmatrix}\omega{\bf x}^{\text{\tiny (n)}}_m\\\omega\end{bmatrix}\,,
\end{equation}
or
\begin{equation}\label{eq03}
\begin{bmatrix}p_{1n} & p_{2n} & p_{3n} & p_{4n}\\
p_{5n} & p_{6n} & p_{7n} & p_{8n}\\
p_{9n} & p_{10n} & p_{11n} & p_{12n}\end{bmatrix}\begin{bmatrix}x_m\\y_m\\z_m\\1\end{bmatrix} = \begin{bmatrix}\omega u_{mn}\\\omega v_{mn}\\\omega\end{bmatrix}\, ,
\end{equation}
where $\omega$ is an unknown scale parameter. Removing the parameter $\omega$, Eq.~\eqref{eq03} yields the following constraints
\begin{equation}\label{eq04}
\begin{split}
p_{1n}x_m &+ p_{2n}y_m + p_{3m}z_m + p_{4n}
= u_{mn}(p_{9n}x_m + p_{10n}y_m + p_{11n}z_m + p_{12n}),\\[1mm]
p_{5n}x_m &+ p_{6n}y_m + p_{7m}z_m + p_{8n}
= v_{mn}(p_{9n}x_m + p_{10n}y_m + p_{11n}z_m + p_{12n})\, .
\end{split}
\end{equation}
These constraints are the basic relationships between the world points and key points via the projection~matrix.

\subsection{Correspondences}
Applications of multi-view geometry can be easily found in computer vision and computer graphics. Most problems in computer vision and computer graphics use images as the first information to exploit. From these images, there are lots of potential methods such as SURF~\citep{Bay2006}, SIFT~\citep{Lowe1999}, PCA-SIFT~\citep{Yan2004}, GLOH~\citep{Mikolajczyk2004} to extract \emph{features} or \emph{key-points}. If ${\bf x}^{\text{\tiny (n)}}_m$ is an image point of the world point ${\bf X}_m$ on the $n$th image for $n = 1,2,\ldots, N$, the sequence of $N$ image points $\{{\bf x}^{\text{\tiny(1)}}_m, {\bf x}^{\text{\tiny(2)}}_m,\ldots, {\bf x}^{\text{\tiny(N)}}_m\}$ is called a \emph{correspondence} and denoted by
\begin{equation}\label{eq05}
\big\{{\bf x}^{\text{\tiny (1)}}_m \,\leftrightarrow\, {\bf x}^{\text{\tiny (2)}}_m \,\leftrightarrow\, \cdots \,\leftrightarrow\, {\bf x}^{\text{\tiny (N)}}_m\big\}\, .
\end{equation}
Sometimes, we use an $N$\emph{-correspondence} with meaning that this correspondence is from $N$ images. An example of the correspondence or three-correspondence is found in Figure~\ref{Fig:Multi_view_geo} from the three blue circles (image points) and one red square (a world point) marked by the black circles. It is easy to see an important property of the correspondence via the centers of cameras (the red triangles) is that all the lines through the point in the correspondence and their centers of cameras will intersect at one point. This point is the  world point corresponding to this correspondence. In three-dimensional space, the probability of two lines intersecting at a point is close to one. But a chance for three, four or more lines intersecting at a point is very close to zero. Hence, this characteristic illuminates a role to have highly accurate estimations of correspondences.

\subsection{$N$-view and $M$-point Camera Calibration}
From $N$ images ${\bf P}_1, {\bf P}_2, \ldots, {\bf P}_N$, using SURF, SIFT or GLOH, we assume that the $N$-correspondences of $M$ world points ${\bf X}_1, {\bf X}_2, \ldots, {\bf X}_M$ are determined. They are
\begin{equation}\label{eq06}
\big\{{\bf x}^{\text{\tiny (1)}}_m \,\leftrightarrow\, {\bf x}^{\text{\tiny (2)}}_m \,\leftrightarrow\, \cdots \,\leftrightarrow\, {\bf x}^{\text{\tiny (N)}}_m\big\}^M_{m=1} \subset \mathbb{R}^2\, ,
\end{equation}
where ${\bf x}^{\text{\tiny (n)}}_m = [u_{mn},v_{mn}]^T$. When the correspondences are acquired, a lot of researchers in the computer vision naturally have had a big attention on the estimation of the world points or the projection matrices based on their correspondences. This problem is called the \emph{$N$-view and $M$-point camera calibration} and it is a prerequisite in many applications of computer vision.

Formally, the $N$-view and $M$-point camera calibration is stated as follows. Given $M$ correspondences $\{{\bf x}^{\text{\tiny(1)}}_m\leftrightarrow{\bf x}^{\text{\tiny(2)}}_m\leftrightarrow\cdots\leftrightarrow{\bf x}^{\text{\tiny (N)}}_m\}^M_{m=1}$, finding $M$ world points $\{{\bf X}_m\}^M_{m=1}$ and $N$ projection matrices $\{{\bf P}_n\}^N_{n=1}$ such that Eq.~\eqref{eq02} or Eq.~\eqref{eq04} satisfy for all $m,n$. Then we call the group $\{{\bf X}_m,{\bf P}_n\}$ is a \emph{solution} of the $N$-view and $M$-point with respect to $\{{\bf x}^{\text{\tiny(1)}}_m\leftrightarrow{\bf x}^{\text{\tiny(2)}}_m\leftrightarrow\cdots\leftrightarrow{\bf x}^{\text{\tiny(N)}}_m\}$.

On finding the solutions for the $N$-view, $M$-point camera calibration, we shall notice about its ambiguity which is explained as follows. Consider the \emph{projective transformation} ${\bf H}$ on the three-dimensional space, where  ${\bf H}$ is a non-singular $4\times 4$ matrix with 15 degrees of freedom~\citep{Hartley2004}. We let the new world points $\widehat{\bf X}_m$ and the new projection matrices $\widehat{\bf P}_n$ be
\begin{equation}\label{eq07}
\begin{bmatrix}\widehat{\bf X}_m\\1\end{bmatrix} = {\bf H}\begin{bmatrix}{\bf X}_m\\ 1\end{bmatrix}\quad \text{and} \quad \widehat{\bf P}_n = {\bf P}_n{\bf H}^{-1},
\end{equation}
for all $m,n$. This means that
\begin{equation}\label{eq08}
\widehat{\bf P}_n\begin{bmatrix}\widehat{\bf X}_m\\1\end{bmatrix} = {\bf P}_n\begin{bmatrix}{\bf X}_m\\1\end{bmatrix}\,
\end{equation}
Thus, if $\{{\bf X}_m, {\bf P}_n\}$ is a solution w.r.t. $\{{\bf x}^{\text{\tiny (1)}}_m\leftrightarrow{\bf x}^{\text{\tiny (2)}}_m\leftrightarrow\cdots\leftrightarrow{\bf x}^{\text{\tiny(N)}}_m\}$, $\{\widehat{\bf X}_m,\widehat{\bf P}_n\}$ is also another solution w.r.t. $\{{\bf x}^{\text{\tiny (1)}}_m\leftrightarrow{\bf x}^{\text{\tiny (2)}}_m\leftrightarrow\cdots\leftrightarrow{\bf x}^{\text{\tiny(N)}}_m\}$. In other words, the $N$-view, $M$-point camera calibration is invariant under a projective transformation.

\subsection{Accuracy of Correspondences}

Correspondences will be the input of the $N$-view and $M$-point camera calibration. Naturally, we have an interesting question about an effect of the correspondences estimation on the solutions of the $N$-view and $M$-point camera calibration. To answer this question we study the $N$-view and $M$-point problem from a small real experiment on reconstructing the Buddha statue. As shown in Figure~\ref{Fig:4View13Points}, there are thirteen world points on the Buddha's face and they are marked by some special symbols. Distances between two of them are manually measured carefully, and from that their relative positions are determined by using the \emph{multidimensional scale} (MDS)~\citep{Torgerson1952,Schonemann1970}. Hence, ground truths for the thirteen world points are acquired. The Buddha's face, together with the thirteen world points, are captured by four images. Since the world points are marked by the special symbols, their correspondences are easily determined with high accuracy. They are indicated by the \emph{red plus symbols} as shown in Figure~\ref{Fig:4View13Points}. From the ground truths for the world points and the image points, when the number of the world points capturing by the $n$th image, $M_0$, is at least six, the projection matrix of this image is derived by the following~linear~equation
\begin{equation}\label{eq09}
{\bf K}^T\begin{bmatrix}p_{1n}\\p_{2n}\\\vdots\\p_{12n}\end{bmatrix}_{12\times 1} \hspace{-2mm}=~~ {\bf 0}_{2M_0\times 1\, ,}
\end{equation}
where ${\bf K}$ is a $12\times2M_0$ matrix constructed by the coordinates of the world points ($x_m, y_m, z_m$) and the image points ($u_{mn}, v_{mn}$) as follows
\begin{align}\label{eq10}
{\bf K}\defeq \begin{bmatrix}x_1 & 0 & \ldots & x_{M_0} & 0\\
                             y_1 & 0 & \ldots & y_{M_0} & 0\\
                             z_1 & 0 & \ldots & z_{M_0} & 0\\
                             1 & 0 & \ldots & 1 & 0\\
                             0 & x_1 & \ldots & 0 & x_{M_0}\\
                             0 & y_1 & \ldots & 0 & y_{M_0}\\
                             0 & z_1 & \ldots & 0 & z_{M_0}\\
                             0 & 1 & \ldots & 0 & 1\\
                             -u_{1n}x_1 & -v_{1n}x_1 & \ldots & -u_{M_0n}x_{M_0} & -v_{M_0n}x_{M_0}\\
                             -u_{1n}y_1 & -v_{1n}y_1 & \ldots & -u_{M_0n}y_{M_0} & -v_{M_0n}y_{M_0}\\
                             -u_{1n}z_1 & -v_{1n}z_1 & \ldots & -u_{M_0n}z_{M_0} & -v_{M_0n}z_{M_0}\\
                             -u_{1n} & -v_{1n} & \ldots & -u_{M_0n} & -v_{M_0n}\end{bmatrix}_{12\times 2M_0\, .}
\end{align}

\begin{figure}[t]
  \centering
  \includegraphics[width=3.5in]{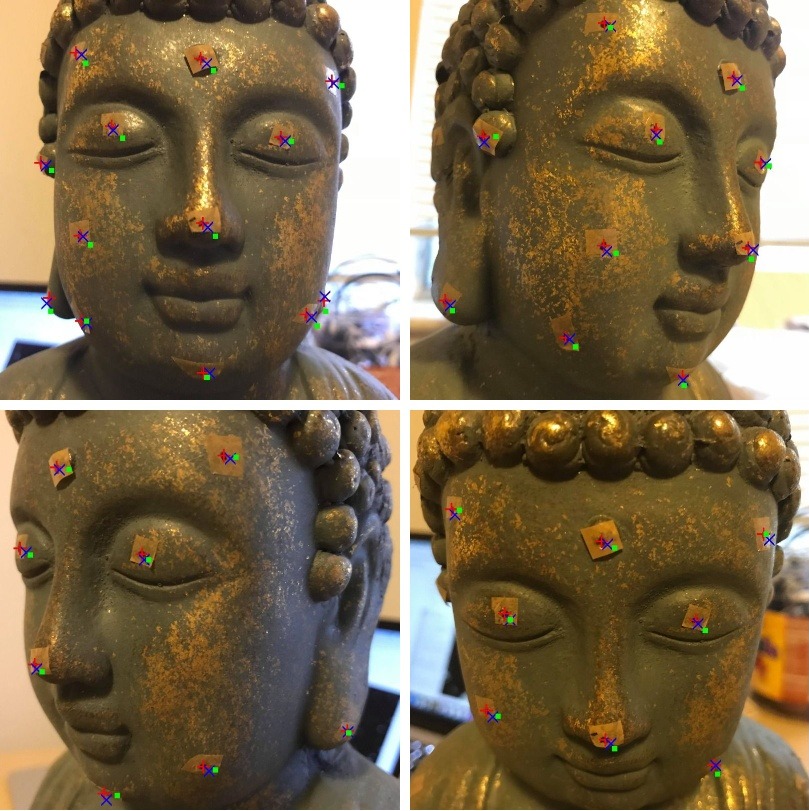}
  \caption{An example of the four-view and thirteen-point camera calibration. The \emph{red plus symbols} are the ground truths, the \emph{blue cross symbols} and \emph{green squares} are estimations of the image points with small and big noises respectively.}
  \label{Fig:4View13Points}
\end{figure}

\begin{figure}[b!]
  \centering
  \includegraphics[width=3.5in]{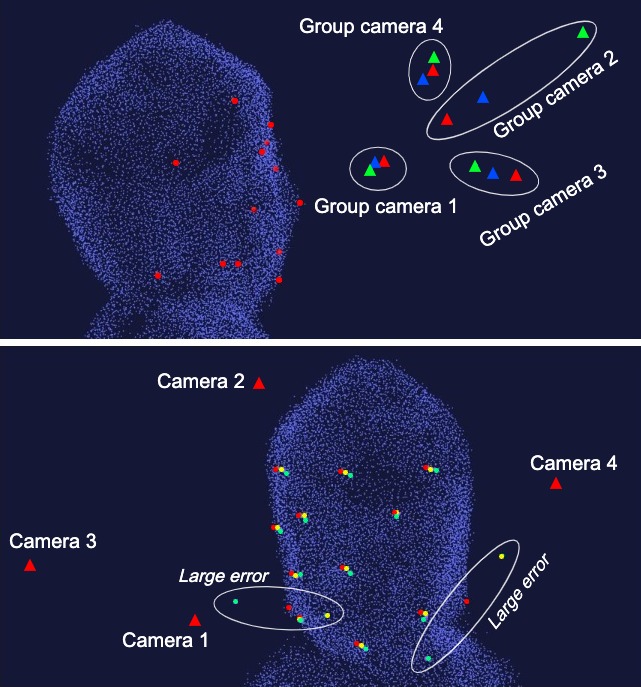}
  \caption{Estimations of the centers of cameras and the world points corresponding to the ground truths and estimated image points with small and big noises in Figure~\ref{Fig:4View13Points}.}
  \label{Fig:Accuracy}
\end{figure}

The centers of cameras are derived from their projection matrices~\citep{Hartley2004}. Summarily, we emphasize that with this experiment the ground truths for the world points, correspondences and centers of cameras are determined with high accuracy. The ground truths for the thirteen world points and four centers of camera are simulated by the \emph{red circles} and \emph{red triangles} in Figure~\ref{Fig:Accuracy}.

Next we study the effect of the correspondences estimation on the $N$-view and $M$-point problem by corrupting the ground truths of the correspondences and seeing how the world points and centers of cameras change. Two different level noises (small and big) are used to corrupt the image points. The first corruption with small noises changes the ground truths indicated by the \emph{red plus symbols} to the estimated image points indicated by the \emph{blue cross symbols}. The second one with big noises gives other estimated image points indicated by \emph{green squares}. Figure~\ref{Fig:4View13Points} simulates clearly the differences among the ground truths (red plus), estimated image points with small noises (blue cross) and estimated image points with big noises (green squares) on each image.

The above sub-figure of Figure~\ref{Fig:Accuracy} presents the ground truth and estimated centers of cameras for the two above corruptions. The red, blue and green triangles are the ground truths, estimations with small noises and estimations with big noises, respectively. To obtain the blue triangles for the estimated centers of cameras with small noises, we use the linear equation Eq.~\eqref{eq09} in which the matrix ${\bf K}$ is constructed by the ground truths of the world points and the estimated image points indicated by the blue cross in each images. Similarly, we replace the blue estimated image points to the green ones on constructing ${\bf K}$ to get the estimations of the green triangles corresponding to the big noises corruption. The below sub-figure presents the ground truths and estimated world points. The red, yellow and green dots are the ground truths, estimated world points with small noises and estimated one with big noises, respectively. The yellow dots are determined based on the blue cross image points and the ground truths of the centers of cameras that are indicated by the red triangles. Formally, a formula for the world point ${\bf X}_m = [x_m,y_m,z_m]^T$ from the $N$ images $\{{\bf P}_n\}$ and $N$ image points $\{{\bf x}^{\text{\tiny (n)}}_m = [u_{mn},v_{mn}]^T\}$ given by
\begin{equation}\label{eq11}
{\bf H}^T\begin{bmatrix}x_m\\y_m\\z_m\\1\end{bmatrix}_{4\times 1}\hspace{-2mm}=~~{\bf 0}_{2N\times 1\, ,}
\end{equation}
where ${\bf H}$ is the following $4\times 2N$ matrix
\begin{align}\label{eq12}
{\bf H}\defeq
\begin{bmatrix}
p_{1,1}-u_{m1}p_{9,1} & p_{5,1}-v_{m1}p_{9,1} & \ldots & p_{5N}-v_{mN}p_{9N}\\[1mm]
p_{2,1}-u_{m1}p_{10,1} & p_{6,1}-v_{m1}p_{10,1} & \ldots & p_{6N}-v_{mN}p_{10N}\\[1mm]
p_{3,1}-u_{m1}p_{11,1} & p_{7,1}-v_{m1}p_{11,1} & \ldots & p_{7N}-v_{mN}p_{11N}\\[1mm]
p_{4,1}-u_{m1}p_{12,1} & p_{8,1}-v_{m1}p_{12,1} & \ldots & p_{8N}-v_{mN}p_{12N}
\end{bmatrix}_{4\times 2N\, .}
\end{align}
The green dots are determined similarly as the yellow ones when we replace the blue cross estimated image points by the green square estimations.

All four images in Figure~\ref{Fig:4View13Points} have the size $3024\times 3024$ pixels. The average differences between the red pluses (ground truths) and the blue crosses (the estimations with small noises) is around 40 pixels, and the similar average between the red pluses and the green squares is around 80 pixels. Intuitively, these differences are small but the errors on estimating the second center of camera and the two world points located on two ears of the Budhha's face are too large. Note that the solutions presented in Figure~\ref{Fig:Accuracy} are found within some extra information, i.e., the centers of cameras are found based on knowing the world points, and vice versa the world points are determined based on knowing the centers of cameras. So, we cannot expect better solutions for the $N$-view and $M$-motion camera calibration if the noises on the estimations of correspondences are similar as Figure~\ref{Fig:4View13Points}. Indeed, to solve some potential applications of the $N$-view and $M$-point problem such as the near filed 3D reconstruction, we should have better results than ones shown in Figure~\ref{Fig:4View13Points} and Figure~\ref{Fig:Accuracy}.

Some recent realistic applications in computer vision~\citep{Hartley2004,Zheng2017,Zollhofer2018} show that we can get good solutions for the $N$-view and $M$-point problem when we (i) exploit more information from the cameras, (ii) improve SURF, SIFT, etc. to get better features or interested points, or/and (iii) increase the number of images. In this paper, we shall propose a different method. Given the correspondences, we shall refine them before applying some existent algorithms to derive the estimations of the world points and the centers of cameras. Our refinement only bases on the geometrical characteristics of the world points, the correspondences and the centers of cameras in the multi-view geometry. This is an advantage and from that we can combine our method with other methods to get better results. This idea that improving the observations based on their algebraic or geometrical properties can be found in~\citep{Velasco2016,Kien2018} on studying the signal-based localization.

\section{Rank Properties for Correspondences}\label{sec:rank_properties}
This section presents the main contribution of the paper. Depending on the number of views $N$ and the number of key-points $M$, an upper bound rank for a matrix constructed by the correspondences is derived. The upper bound ranks of these matrices give us some linear constraints among the elements of these matrices. Thus, if these elements are simple functions of the correspondences, these linear constraints will help us on recognizing the outliers and refine the correspondences.

The first contribution of the paper is in Theorem~\ref{Thm:RanktwoView} from that a \emph{rank deficient $M\times 9$ matrix} is constructed from $M$ correspondences of two images with $M\geq 9$. Because of the rank deficient property, there is a linear constraint among this matrix's elements. Interestingly, since all the four times $M$ coordinates of the $M$ correspondences are in these elements, we shall use the rank deficient property given by Theorem~\ref{Thm:RanktwoView} to refine these $4M$ coordinates. Theorem~\ref{Thm:RankfiveView}, proposing another \emph{rank deficient $N\times 5$ matrix} constructed from six correspondences of $N$ images with $N\geq 5$, is the second contribution of the paper. Unlike Theorem~\ref{Thm:RanktwoView}, the elements of the $N\times 5$ matrix in Theorem~\ref{Thm:RankfiveView} are quartic functions of $12N$ coordinates from six correspondences in $N$ images. Thus, it is difficult to get a good algorithm to refine these correspondences based on the rank deficient property given by Theorem~\ref{Thm:RankfiveView}. However, as shown in the next section, this rank deficient property is efficient to \emph{outliers refinement} and \emph{self-estimation} in correspondences estimations.

The result proposed by Theorem~\ref{Thm:RanktwoView} is correct for all two-view and $M$-point with $M \geq 9$, while the one in Theorem~\ref{Thm:RankfiveView} is correct for all $N$-view and six-point with $N\geq 5$. Hence when the number of correspondences is large but the number of images is small, we shall use Theorem~\ref{Thm:RanktwoView} in recognizing the outliers and refining the correspondences. Vice versa, when the number of correspondences is small and the number of images is larger, Theorem~\ref{Thm:RankfiveView} can be used. Note that the results in Theorem~\ref{Thm:RanktwoView} and Theorem~\ref{Thm:RankfiveView} cannot work for the cases of the number of correspondences be smaller than nine and the number of images be smaller than five.

Before introducing the two main theorems we mention that we shall use some simple notations for the cases of two and three-view. More precisely, in the two-view case we shall use $\{{\bf x}_m\leftrightarrow{\bf x}'_m\}$ and $\{{\bf P}, {\bf Q}\}$ to denote the correspondence and the projection matrices, and in the three-view case we use $\{{\bf x}_m\leftrightarrow{\bf x}'_m\leftrightarrow{\bf x}''_m\}$ and $\{{\bf P}, {\bf Q}, {\bf R}\}$.

\subsection{Rank Property for the Two-View}
The first result of this paper for the two-view and $M$-point camera calibration is given by the following theorem.
\begin{theo}\label{Thm:RanktwoView}
Let
\begin{equation}\label{eq13}
\bigg\{{\bf x}_m \defeq \begin{bmatrix}u_m\\v_m\end{bmatrix}~\leftrightarrow~{\bf x}'_m\defeq\begin{bmatrix}u'_m\\v'_m\end{bmatrix}\bigg\}^M_{m=1}
\end{equation}
be $M$ correspondences on two images. Then the following $M\times 9$ matrix ${\bf \Gamma}$ given by
\begin{equation}\label{eq14}
\begin{bmatrix}
1 & u_1 & v_1 & u'_1 & v'_1 & u_1u'_1 & u_1v'_1 & v_1u'_1 & v_1v'_1\\[1mm]
1 & u_2 & v_2 & u'_2 & v'_2 & u_2u'_2 & u_2v'_2 & v_2u'_2 & v_2v'_2\\[1mm]
1 & u_3 & v_3 & u'_3 & v'_3 & u_3u'_3 & u_3v'_3 & v_3u'_3 & v_3v'_3\\[1mm]
\vdots & \vdots & \vdots & \vdots & \vdots & \vdots & \vdots & \vdots & \vdots\\[1mm]
1 & \hspace{-0.5mm}u_M\hspace{-0.5mm} & \hspace{-0.5mm}v_M\hspace{-0.5mm} & \hspace{-0.5mm}u'_M\hspace{-0.5mm} & \hspace{-0.5mm}v'_M\hspace{-0.5mm} & \hspace{-0.5mm}u_Mu'_M\hspace{-0.5mm} & \hspace{-0.5mm}u_Mv'_M\hspace{-0.5mm} & \hspace{-0.5mm}v_Mu'_M\hspace{-0.5mm} & \hspace{-0.5mm}v_Mv'_M
\end{bmatrix}_{M\times 9}
\end{equation}
will have the rank at most 8.
\end{theo}
\begin{proof}
Let $\{{\bf X}_m = [x_m,y_m,z_m]^T\}^M_{m=1}$ and
\begin{equation}\label{eq15}
{\bf P} = \begin{bmatrix}
p_1 & p_2 & p_3 & p_4\\
p_5 & p_6 & p_7 & p_8\\
p_9 & p_{10} & p_{11} & p_{12}
\end{bmatrix}_{3\times 4}~,\qquad~{\bf Q} = \begin{bmatrix}
q_1 & q_2 & q_3 & q_4\\
q_5 & q_6 & q_7 & q_8\\
q_9 & q_{10} & q_{11} & q_{12}
\end{bmatrix}_{3\times 4}
\end{equation}
be the $M$ world points and two projection matrices corresponding to the two-correspondence $\{{\bf x}_m\leftrightarrow{\bf x}'_m\}$. The equation Eq.~\eqref{eq11} which is from Eq.~\eqref{eq04}, implies
\begin{equation}\label{eq16}
{\bf K}\begin{bmatrix}x_m\\y_m\\z_m\\1\end{bmatrix} = {\bf 0}_{4\times 1}
\end{equation}
where
\begin{equation}\label{eq17}
{\bf K} = \begin{bmatrix}
p_1 - u_{m}p_9 & p_2 - u_{m}p_{10} & p_3 - u_{m}p_{11} & p_4 - u_{m}p_{12}\\[1mm]
p_5 - v_{m}p_9 & p_6 - v_{m}p_{10} & p_7 - v_{m}p_{11} & p_8 - v_{m}p_{12}\\[1mm]
q_1 - u'_{m}q_9 & q_2 - u'_{m}q_{10} & q_3 - u'_{m}q_{11} & q_4 - u'_{m}p_{12}\\[1mm]
q_5 - v'_{m}q_9 & q_6 - v'_{m}q_{10} & q_7 - v'_{m}q_{11} & q_8 - v'_{m}p_{12}
\end{bmatrix}_{4\times 4\, .}
\end{equation}
It is clear from Eq.~\eqref{eq17} that the determinant of ${\bf K}$ should be zero, i.e., $\text{det}({\bf K}) = 0$. Considering $u_{m}, v_{m}, u'_{m}, v'_{m}$ as unknown variables and $\{p_i, q_j\}$ as coefficients, we present the determinant of ${\bf K}$ as a polynomial equation that
\begin{equation}\label{eq18}
\begin{split}
\text{det}({\bf K}) &= a_0 + a_1u_{m} + a_2v_{m} + a_3u'_{m} + a_4v'_{m}\\
&~~~+a_5u_{m}v_{m} + a_6u_{m}u'_{m} + a_7u_{m}v'_{m} + a_8v_mu'_m + a_9v_mv'_m + a_{10}u'_{m}v'_{m}\\
&~~~+a_{11}u_{m}v_{m}u'_{m} + a_{12}u_mv_mv'_m + a_{13}u_mu'_mv'_m + a_{14}v_{m}u'_{m}v'_{m}\\
&~~~+a_{15}u_{m}v_{m}u'_{m}v'_{m}\, .
\end{split}
\end{equation}
This polynomial has sixteen monomials and they are
\begin{equation}\label{eq19}
\begin{split}
&1,~ u_{m},~ v_{m},~ u'_{m},~ v'_{m},~u_{m}v_{m},~ u_{m}u'_{m},~ u_{m}v'_{m},~v_{m}u'_{m},~v_{m}v'_{m},~ u'_{m}v'_{m},\\
&u_{m}v_{m}u'_{m},~ u_{m}v_{m}v'_{m},~u_{m}u'_{m}v'_{m},~v_{m}u'_{m}v'_{m},~u_{m}v_{m}u'_{m}v'_{m}\, ,
\end{split}
\end{equation}
corresponding with sixteen coefficients $a_0, a_1,\ldots, a_{15}$. Interestingly, we always have
\begin{equation}\label{eq20}
a_5 = a_{10} = a_{11} = a_{12} = a_{13} = a_{14} = a_{15} = 0\, .
\end{equation}
For example, the coefficient $a_5$ is given by
\begin{equation}\label{eq21}
\begin{split}
&p_{9}p_{10}q_{3}q_{8} + p_{9}p_{11}q_{4}q_{6} + p_{9}p_{12}q_{2}q_{7}+ p_{10}p_{12}q_{3}q_{5}
+ p_{10}p_{11}q_{1}q_{8} + p_{10}p_{9}q_{4}q_{7}\\
+\,&p_{11}p_{9}q_{2}q_{8} + p_{11}p_{10}q_{4}q_{5} + p_{11}p_{12}q_{1}q_{6} + p_{12}p_{11}q_{2}q_{5} + p_{12}p_{10}q_{1}q_{7} + p_{12}p_{9}q_{3}q_{6}\\
-\,&p_{9}p_{12}q_{3}q_{6} - p_{9}p_{11}q_{2}q_{8} - p_{9}p_{10}q_{4}q_{7} - p_{10}p_{9}q_{3}q_{8} - p_{10}p_{11}q_{4}q_{5} - p_{10}p_{12}q_{1}q_{7}\\
-\,&p_{11}p_{12}q_{2}q_{5} - p_{11}p_{10}q_{1}q_{8} - p_{11}p_{9}q_{4}q_{6} - p_{12}p_{9}q_{2}q_{7} - p_{12}p_{10}q_{3}q_{5} - p_{12}p_{11}q_{1}q_{6}\\
=~&0\, .\\[-3mm]
\end{split}
\end{equation}
Thus, the polynomial in Eq.~\eqref{eq18} with 16 monomials will be reduced to 9 monomials, and it is
\begin{equation}\label{eq22}
\begin{bmatrix}
1\\u_{m}\\v_{m}\\u'_{m}\\v'_{m}\\u_{m}u'_{m}\\u_{m}v'_{m}\\v_{m}u'_{m}\\v_{m}v'_{m}
\end{bmatrix}^T\begin{bmatrix}
a_0\\a_1\\a_2\\a_3\\a_4\\a_6\\a_7\\a_8\\a_9
\end{bmatrix} ~=~ 0\, .
\end{equation}
Eq.~\eqref{eq22} is correct for all $m = 1,2,\ldots,M$. Therefore,
\begin{equation}\label{eq23}
\begin{bmatrix}
1 & 1 & 1 & \ldots & 1\\
u_1 & u_2 & u_3 & \ldots & u_M\\
v_1 & v_2 & v_3 & \ldots & v_M\\
u'_1 & u'_2 & u'_3 & \ldots & u'_M\\
v'_1 & v'_2 & v'_3 & \ldots & v'_M\\
u_1u'_1 & u_2u'_2 & u_3u'_3 & \ldots & u_Mu'_M\\
u_1v'_1 & u_2v'_2 & u_3v'_3 & \ldots & u_Mv'_M\\
v_1u'_1 & v_2u'_2 & v_3u'_3 & \ldots & v_Mu'_M\\
v_1v'_1 & v_2v'_2 & v_3v'_3 & \ldots & v_Mv'_M
\end{bmatrix}^T\begin{bmatrix}
a_0\\a_1\\a_2\\a_3\\a_4\\a_6\\a_7\\a_8\\a_9
\end{bmatrix} ~=~ 0\, .
\end{equation}
Since the vector $[a_0, a_1, \ldots, a_9]^T$ is non-zero, the proof of Theorem~\ref{Thm:RanktwoView} is finished by Eq.~\eqref{eq23}.
\end{proof}

Theorem~\ref{Thm:RanktwoView} confirms that nine columns of the matrix ${\bf \Gamma}$ in Eq.~\eqref{eq14} are linearly dependent. When the number of key-points is at least nine ($M\geq 9$), the columns linear independence yields the rank deficient property for ${\bf \Gamma}$. The key point of the proof of Theorem~\ref{Thm:RanktwoView} is the linear relationship Eq.~\eqref{eq22} from that all the coordinates of the $m$th point correspondences are in the first vector and all the elements of the two projection matrices are in the second vector. Thus this linear relationship is still correct when we replace the first vector by another from other correspondences. In another word, Theorem~\ref{Thm:RanktwoView} is correct not only for nine two-correspondences but also ten, eleven or more two-correspondences.

When $M \geq 9$, considering the $9\times 9$ sub-matrix ${\bf \Gamma}_{1,2,\ldots,9}$ constructed by the first nine rows of ${\bf \Gamma}$ in Eq.~\eqref{eq14}, Theorem~\ref{Thm:RanktwoView} guarantees
\begin{equation}\label{eq24}
\text{det}\big({\bf \Gamma}_{1,2,\ldots,9}\big) = 0\, .
\end{equation}
Let us denote ${\bf \Gamma}_{1,2,\ldots,9}^{i}$ the $8\times 8$ matrix from ${\bf \Gamma}_{1,2,\ldots,9}$ when we remove the first row and the $i$-th column $(i = 1,2,\ldots, 9)$. Using the first row to compute the determinant of the $9\times 9$ matrix ${\bf \Gamma}_{1,2,\ldots,9}$ via its $8\times 8$ matrices, Eq.~\eqref{eq24} implies
\begin{equation}\label{eq25}
\alpha_{1|2,\ldots,9}u_{1} + \beta_{1|2,\ldots,9}v_{1} + \gamma_{1|2,\ldots,9} = 0\, ,
\end{equation}
where
\begin{equation}\label{eq26}
\begin{split}
    \alpha_{1|2,\ldots,9} &~\defeq~ -\big|{\bf \Gamma}_{1,2,\ldots,9}^2\big| - u'_{1}\big|{\bf \Gamma}_{1,2,\ldots,9}^6\big| + v'_{1}\big|{\bf \Gamma}_{1,2,\ldots,9}^7\big|\\[2mm]
    \beta_{1|2,\ldots,9} &~\defeq~ \big|{\bf \Gamma}_{1,2,\ldots,9}^3\big| - u'_{1}\big|{\bf \Gamma}_{1,2,\ldots,9}^8\big| + v'_{1}\big|{\bf \Gamma}_{1,2,\ldots,9}^9\big|\\[2mm]
    \gamma_{1|2,\ldots,9} &~\defeq~ \big|{\bf \Gamma}_{1,2,\ldots,9}^1\big| - u'_1\big|{\bf \Gamma}_{1,2,\ldots,9}^4\big| + v'_1\big|{\bf \Gamma}_{1,2,\ldots,9}^5\big|\, .
\end{split}
\end{equation}

The linear equation Eq.~\eqref{eq25} brings us an idea that if the image point ${\bf x}_1 = [u_1,v_1]^T$ is an outlier and other image points are good estimations then the outlier $[u_1,v_1]^T$ can be refined or recomputed based on Eq.~\eqref{eq25}. Only one linear equation as Eq.~\eqref{eq25} is not sufficient to determine two variables $u_1$ and $v_1$. Fortunately, when $M\geq 10$, we have at least these nine linear equations and they are sufficient to give us the good solutions for the outlier $[u_1,v_1]^T$. Formally, this idea is solved by the following proposition.
\begin{prop}\label{Prop:SelftwoView}
Assuming that $M\geq 10$, the functions of $u_1$ and $v_1$ from the remain coordinates are given by
\begin{equation}\label{eq27}
\begin{bmatrix}u_1\\v_1\end{bmatrix} = -\begin{bmatrix}\alpha_{1|2,\ldots,9} & \beta_{1|2,\ldots,9}\\
\vdots & \vdots\\\alpha_{1|i_1,\ldots,i_8} & \beta_{1|i_1,\ldots,i_8}\\\vdots & \vdots\\\alpha_{1|M-7,\ldots,M} & \beta_{1|M-7,\ldots,M}\end{bmatrix}^{\dagger}\begin{bmatrix}\gamma_{1|2,\ldots,9}\\\vdots\\\gamma_{1|i_1,\ldots,i_8}\\\vdots\\\gamma_{1|M-7,\ldots,M}\end{bmatrix}
\end{equation}
for all $2\leq i_1 < i_2 < \cdots < i_8 \leq M$, where $\dagger$ denotes Moore-Penrose inverse operator.
\end{prop}
\begin{proof}
Eq.~\eqref{eq25} is correct for all nine indices $1$ and $2\leq i_1 < i_2 < \cdots < i_8 \leq M$. So,
\begin{equation}\label{eq28}
\begin{bmatrix}\alpha_{1|i_1,i_2,\ldots,i_8} & \beta_{1|i_1,i_2,\ldots,i_8}\end{bmatrix}\begin{bmatrix}u_1\\v_1\end{bmatrix} = -\gamma_{1|i_1,i_2,\ldots,i_8}\, .
\end{equation}
Hence, let eight indices $(i_1,\ldots,i_8)$ run from the first group $(2,\ldots,9)$ to the last group $(M-7,\ldots,M)$ we get
\begin{equation}\label{eq29}
\begin{bmatrix}\alpha_{1|2,\ldots,9} & \beta_{1|2,\ldots,9}\\\vdots & \vdots\\\alpha_{1|i_1,i_2,\ldots,i_8} & \beta_{1|i_1,i_2,\ldots,i_8}\\\vdots & \vdots\\\alpha_{1|M-7,\ldots,M} & \beta_{1|M-7,\ldots,M}\end{bmatrix}\begin{bmatrix}u_1\\v_1\end{bmatrix} = -\begin{bmatrix}\gamma_{1|2,\ldots,9}\\\vdots\\\gamma_{1|i_1,i_2,\ldots,i_8}\\\vdots\\\gamma_{1|M-7,\ldots,M}\end{bmatrix}\, .
\end{equation}
Note that the matrix in the left side of Eq.~\eqref{eq29} has the size ${M-1 \choose 8}\times 2$ and its is always full rank. Thus, Eq.~\eqref{eq27} is as a consequence of Eq.~\eqref{eq29} and Proposition~\ref{Prop:SelftwoView} is correct.
\end{proof}

\subsection{Rank Property for the Multi-View}
To get another rank property for the multi-view case we follow the idea from the linear constraint Eq.~\eqref{eq22} to derive another linear constraint that the first vector is from the coordinates of the correspondences and the second vector is from the coordinates of the world points. Thus, because the second vector is  independent to the images, if this linear constraint is found it will be correct for all images. Then an upper bound rank for the multi-view can be obtained.

Theorem~\ref{Thm:RankfiveView}, basing on the above discussion, proposes an upper bound rank for an $N\times 5$ matrix constructed from the six point correspondences in $N$ images.
\begin{theo}\label{Thm:RankfiveView}
Let
\begin{equation}\label{eq30}
\Big\{{\bf x}_m^{\texttt{\tiny(1)}} \,\leftrightarrow\, {\bf x}_m^{\texttt{\tiny(2)}} \,\leftrightarrow\, {\bf x}_m^{\texttt{\tiny(3)}} \,\leftrightarrow\, \cdots \,\leftrightarrow\, {\bf x}_m^{\texttt{\tiny(N)}}\Big\}^6_{m=1}
\end{equation}
be six correspondences on $N$ images, where ${\bf x}^{\texttt{\tiny(n)}}_m\defeq[u_{mn},v_{mn}]^T$. The following $N\times 5$ matrix
\begin{equation}\label{eq31}
{\bf \Lambda} ~\defeq~\begin{bmatrix}
\xi^{\text{\tiny(1)}}_{34,5}\xi^{\text{\tiny(1)}}_{12,6} & \xi^{\text{\tiny(1)}}_{42,5}\xi^{\text{\tiny(1)}}_{13,6} & \xi^{\text{\tiny(1)}}_{23,5}\xi^{\text{\tiny(1)}}_{14,6} &
\xi^{\text{\tiny(1)}}_{12,5}\xi^{\text{\tiny(1)}}_{34,6} & \xi^{\text{\tiny(1)}}_{13,5}\xi^{\text{\tiny(1)}}_{42,6}\\[3mm]
\xi^{\text{\tiny(2)}}_{34,5}\xi^{\text{\tiny(2)}}_{12,6} & \xi^{\text{\tiny(2)}}_{42,5}\xi^{\text{\tiny(2)}}_{13,6} & \xi^{\text{\tiny(2)}}_{23,5}\xi^{\text{\tiny(2)}}_{14,6} &
\xi^{\text{\tiny(2)}}_{12,5}\xi^{\text{\tiny(2)}}_{34,6} & \xi^{\text{\tiny(2)}}_{13,5}\xi^{\text{\tiny(2)}}_{42,6}\\[3mm]
\xi^{\text{\tiny(3)}}_{34,5}\xi^{\text{\tiny(3)}}_{12,6} & \xi^{\text{\tiny(3)}}_{42,5}\xi^{\text{\tiny(3)}}_{13,6} & \xi^{\text{\tiny(3)}}_{23,5}\xi^{\text{\tiny(3)}}_{14,6} &
\xi^{\text{\tiny(3)}}_{12,5}\xi^{\text{\tiny(3)}}_{34,6} & \xi^{\text{\tiny(3)}}_{13,5}\xi^{\text{\tiny(3)}}_{42,6}\\[1mm]
\vdots & \vdots & \vdots & \vdots & \vdots\\[1.5mm]
\xi^{\text{\tiny(N)}}_{34,5}\xi^{\text{\tiny(N)}}_{12,6} & \xi^{\text{\tiny(N)}}_{42,5}\xi^{\text{\tiny(N)}}_{13,6} & \xi^{\text{\tiny(N)}}_{23,5}\xi^{\text{\tiny(N)}}_{14,6} &
\xi^{\text{\tiny(N)}}_{12,5}\xi^{\text{\tiny(N)}}_{34,6} & \xi^{\text{\tiny(N)}}_{13,5}\xi^{\text{\tiny(N)}}_{42,6}
\end{bmatrix}_{N\times 5}
\end{equation}
where
\begin{equation}\label{eq32}
\xi^{\texttt{\tiny(n)}}_{i_1i_2,j} \,\defeq\, \begin{vmatrix}u_{i_1n}-u_{jn} & u_{i_2n}-u_{jn}\\[1mm]v_{i_1n}-v_{jn} & v_{i_2n}-v_{jn}\end{vmatrix}
\end{equation}
will have the rank at most 4.
\end{theo}
\begin{proof}
Let ${\bf X}_1 = [x_1,y_1,z_1]^T, {\bf X}_2 = [x_2,y_2,z_2]^T, \ldots, {\bf X}_6 = [x_6,y_6,z_6]^T$ be the six world points and
\begin{equation}\label{eq33}
{\bf P}_n = \begin{bmatrix}
p_{1n} & p_{2n} & p_{3n} & p_{4n}\\
p_{5n} & p_{6n} & p_{7n} & p_{8n}\\
p_{9n} & p_{10n} & p_{11n} & p_{12n}
\end{bmatrix}_{3\times 4}
\end{equation}
($n=1,2,\ldots, N$) be five projection matrices corresponding to the point correspondences given by Eq.~\eqref{eq30}. From Eq.~\eqref{eq09} and Eq.~\eqref{eq10} we know that within at least six world points the projection matrix ${\bf P}_n$ can be derived by a linear equation. Unfortunately, these six world points are unknown. But we can use this linear equation to present each element of ${\bf P}_n$ as a function of these six world points. We begin with the linear equation presenting the relationship between the six world points and the projection matrix via the corresponding image point.

\begin{equation}\label{eq34}
{\left[\begin{tabular}{cccccccccccc}
$x_1$ & $y_1$ & $z_1$ & 1 & 0 & 0 & 0 & 0 & $-u_{1n}x_1$ & $-u_{1n}y_1$ & $-u_{1n}z_1$ & $-u_{1n}$\\[1mm]
0 & 0 & 0 & 0 & $x_1$ & $y_1$ & $z_1$ & 1 & $-v_{1n}x_1$ & $-v_{1n}y_1$ & $-v_{1n}z_1$ & $-v_{1n}$\\[1mm]
$x_2$ & $y_2$ & $z_2$ & 1 & 0 & 0 & 0 & 0 & $-u_{2n}x_2$ & $-u_{2n}y_2$ & $-u_{2n}z_2$ & $-u_{2n}$\\[1mm]
0 & 0 & 0 & 0 & $x_2$ & $y_2$ & $z_2$ & 1 & $-v_{2n}x_2$ & $-v_{2n}y_2$ & $-v_{2n}z_2$ &$-v_{2n}$\\[1mm]
$x_3$ & $y_3$ & $z_3$ & 1 & 0 & 0 & 0 & 0 & $-u_{3n}x_3$ & $-u_{3n}y_3$ & $-u_{3n}z_3$ & $-u_{3n}$\\[1mm]
0 & 0 & 0 & 0 & $x_3$ & $y_3$ & $z_3$ & 1 & $-v_{3n}x_3$ & $-v_{3n}y_3$ & $-v_{3n}z_3$ & $-v_{3n}$\\[1mm]
$x_4$ & $y_4$ & $z_4$ & 1 & 0 & 0 & 0 & 0 & $-u_{4n}x_4$ & $-u_{4n}y_4$ & $-u_{4n}z_4$ & $-u_{4n}$\\[1mm]
0 & 0 & 0 & 0 & $x_4$ & $y_4$ & $z_4$ & 1 & $-v_{4n}x_4$ & $-v_{4n}y_4$ & $-v_{4n}z_4$ & $-v_{4n}$\\[1mm]
$x_5$ & $y_5$ & $z_5$ & 1 & 0 & 0 & 0 & 0 & $-u_{5n}x_5$ & $-u_{5n}y_5$ & $-u_{5n}z_5$ & $-u_{5n}$\\[1mm]
0 & 0 & 0 & 0 & $x_5$ & $y_5$ & $z_5$ & 1 & $-v_{5n}x_5$ & $-v_{5n}y_5$ & $-v_{5n}z_5$ & $-v_{5n}$\\[1mm]
$x_6$ & $y_6$ & $z_6$ & 1 & 0 & 0 & 0 & 0 & $-u_{6n}x_6$ & $-u_{6n}y_6$ & $-u_{6n}z_6$ & $-u_{6n}$\\[1mm]
0 & 0 & 0 & 0 & $x_6$ & $y_6$ & $z_6$ & 1 & $-v_{6n}x_6$ & $-v_{6n}y_6$ & $-v_{6n}z_6$ & $-v_{6n}$\\[1mm]
\end{tabular}\right]_{12\times 12}\begin{bmatrix}p_{1n}\\[1mm]p_{2n}\\[1mm]p_{3n}\\[1mm]p_{4n}\\[1mm]p_{5n}\\[1mm]p_{6n}\\[1mm]p_{7n}\\[1mm]p_{8n}\\[1mm]p_{9n}\\[1mm]p_{10n}\\[1mm]p_{11n}\\[1mm]p_{12n}\end{bmatrix} \,=~ {\bf 0}_{12\times 1\, .}}
\end{equation}
Thanks to the invariant property of the $N$-view, $M$-point camera calibration under any projective transformation, without loss of generality we can assume that
\setcounter{equation}{34}
\begin{equation}\label{eq35}
\begin{bmatrix}{\bf X}_1\hspace{-1mm} & \hspace{-1mm}{\bf X}_2\hspace{-1mm} & \hspace{-1mm}{\bf X}_3\hspace{-1mm} & \hspace{-1mm}{\bf X}_4\end{bmatrix} \hspace{-0.5mm}=\hspace{-0.5mm} \begin{bmatrix}x_1 & \hspace{-0.5mm}x_2\hspace{-0.5mm} & \hspace{-0.5mm}x_3\hspace{-0.5mm} & x_4\\y_1 & \hspace{-0.5mm}y_2\hspace{-0.5mm} & \hspace{-0.5mm}y_3\hspace{-0.5mm} & y_4\\z_1 & \hspace{-0.5mm}z_2\hspace{-0.5mm} & \hspace{-0.5mm}z_3\hspace{-0.5mm} & z_4\end{bmatrix} \hspace{-0.5mm}=\hspace{-0.5mm} \begin{bmatrix}0 & 0 & 0 & 1\\0 & 0 & 1 & 0\\0 & 1 & 0 & 0\end{bmatrix}.
\end{equation}
Hence, Eq.~\eqref{eq34} becomes
\begin{equation}\label{eq36}
{\left[\begin{tabular}{cccccccccccc}
0 & 0 & 0 & 1 & 0 & 0 & 0 & 0 & 0 & 0 & 0 & $-u_{1n}$\\[1mm]
0 & 0 & 0 & 0 & 0 & 0 & 0 & 1 & 0 & 0 & 0 & $-v_{1n}$\\[1mm]
0 & 0 & 1 & 1 & 0 & 0 & 0 & 0 & 0 & 0 & $-u_{2n}$ & $-u_{2n}$\\[1mm]
0 & 0 & 0 & 0 & 0 & 0 & 1 & 1 & 0 & 0 & $-v_{2n}$ &$-v_{2n}$\\[1mm]
0 & 1 & 0 & 1 & 0 & 0 & 0 & 0 & 0 & $-u_{3n}$ & 0 & $-u_{3n}$\\[1mm]
0 & 0 & 0 & 0 & 0 & 1 & 0 & 1 & 0 & $-v_{3n}$ & 0 & $-v_{3n}$\\[1mm]
1 & 0 & 0 & 1 & 0 & 0 & 0 & 0 & $-u_{4n}$ & 0 & 0 & $-u_{4n}$\\[1mm]
0 & 0 & 0 & 0 & 1 & 0 & 0 & 1 & $-v_{4n}$ & 0 & 0 & $-v_{4n}$\\[1mm]
$x_5$ & $y_5$ & $z_5$ & 1 & 0 & 0 & 0 & 0 & $-u_{5n}x_5$ & $-u_{5n}y_5$ & $-u_{5n}z_5$ & $-u_{5n}$\\[1mm]
0 & 0 & 0 & 0 & $x_5$ & $y_5$ & $z_5$ & 1 & $-v_{5n}x_5$ & $-v_{5n}y_5$ & $-v_{5n}z_5$ & $-v_{5n}$\\[1mm]
$x_6$ & $y_6$ & $z_6$ & 1 & 0 & 0 & 0 & 0 & $-u_{6n}x_6$ & $-u_{6n}y_6$ & $-u_{6n}z_6$ & $-u_{6n}$\\[1mm]
0 & 0 & 0 & 0 & $x_6$ & $y_6$ & $z_6$ & 1 & $-v_{6n}x_6$ & $-v_{6n}y_6$ & $-v_{6n}z_6$ & $-v_{6n}$
\end{tabular}\right]_{12\times 12}\begin{bmatrix}p_{1n}\\[1mm]p_{2n}\\[1mm]p_{3n}\\[1mm]p_{4n}\\[1mm]p_{5n}\\[1mm]p_{6n}\\[1mm]p_{7n}\\[1mm]p_{8n}\\[1mm]p_{9n}\\[1mm]p_{10n}\\[1mm]p_{11n}\\[1mm]p_{12n}\end{bmatrix} \,=~ {\bf 0}_{12\times 1\, .}}
\end{equation}
Eq.~\eqref{eq36} is a simple linear equation for us to present $p_{in},~(i=1,2,\ldots,12)$ as functions of the two remain unknown world points ${\bf X}_5 = [x_5,y_5,z_5]^T$ and ${\bf X}_6 = [x_6,y_6,z_6]^T$. We simplify Eq.~\eqref{eq36} by implementing some linear operators on the rows of the matrix given in Eq.~\eqref{eq36} as follows. First, we subtract the first row for the third, fifth, 7th, 9th and 11th rows, and subtract the second row for the forth, sixth, 8th, 10th and 12th rows. After the first step we have a new matrix. With this new matrix, we subtract $x_5$ times the 7th row, $y_5$ times the fifth row and $z_5$ times the third row from the 9th and 11th rows and $x_5$ times the 8th row, $y_5$ times the 6th row and $z_5$ times the forth row from the 10th and 12th rows. Finally we have another but equivalent linear equations given by
\begin{equation}\label{eq37}
{\left[\hspace{-1mm}\begin{tabular}{cccccccccccc}
0 & 0 & 0 & 1 & 0 & 0 & 0 & 0 & 0 & 0 & 0 & $-u_{1n}$\\
0 & 0 & 0 & 0 & 0 & 0 & 0 & 1 & 0 & 0 & 0 & $-v_{1n}$\\
0 & 0 & 1 & 0 & 0 & 0 & 0 & 0 & 0 & 0 & $-u_{2n}$ & $u_{1n}-u_{2n}$\\
0 & 0 & 0 & 0 & 0 & 0 & 1 & 0 & 0 & 0 & $-v_{2n}$ &$v_{1n}-v_{2n}$\\
0 & 1 & 0 & 0 & 0 & 0 & 0 & 0 & 0 & $-u_{3n}$ & 0 & $u_{1n}-u_{3n}$\\
0 & 0 & 0 & 0 & 0 & 1 & 0 & 0 & 0 & $-v_{3n}$ & 0 & $v_{1n}-v_{3n}$\\
1 & 0 & 0 & 0 & 0 & 0 & 0 & 0 & $-u_{4n}$ & 0 & 0 & $u_{1n}-u_{4n}$\\
0 & 0 & 0 & 0 & 1 & 0 & 0 & 0 & $-v_{4n}$ & 0 & 0 & $v_{1n}-v_{4n}$\\[2mm]
0 & 0 & 0 & 0 & 0 & 0 & 0 & 0 & \hspace{-1mm}$x_5(u_{4n}\hspace{-0.5mm}-\hspace{-0.5mm}u_{5n})$\hspace{-1mm} & \hspace{-1mm}$y_5(u_{3n}\hspace{-0.5mm}-\hspace{-0.5mm}u_{5n})$\hspace{-1mm} & \hspace{-1mm}$z_5(u_{2n}\hspace{-0.5mm}-\hspace{-0.5mm}u_{5n})$\hspace{-1mm} & \hspace{-1mm}$\begin{matrix}(u_{1n}\hspace{-0.5mm}-\hspace{-0.5mm}u_{5n})\hspace{-0.5mm}-\hspace{-0.5mm}x_5(u_{1n}\hspace{-0.5mm}-\hspace{-0.5mm}u_{4n})\\-y_5(u_{1n}\hspace{-0.5mm}-\hspace{-0.5mm}u_{3n})\\-z_5(u_{1n}\hspace{-0.5mm}-\hspace{-0.5mm}u_{2n})\end{matrix}$\hspace{-1mm}\\[6mm]
0 & 0 & 0 & 0 & 0 & 0 & 0 & 0 & \hspace{-1mm}$x_5(v_{4n}\hspace{-0.5mm}-\hspace{-0.5mm}v_{5n})$\hspace{-1mm} & \hspace{-1mm}$y_5(v_{3n}\hspace{-0.5mm}-\hspace{-0.5mm}v_{5n})$\hspace{-1mm} & \hspace{-1mm}$z_5(v_{2n}\hspace{-0.5mm}-\hspace{-0.5mm}v_{5n})$\hspace{-1mm} & \hspace{-1mm}$\begin{matrix}(v_{1n}\hspace{-0.5mm}-\hspace{-0.5mm}v_{5n})\hspace{-0.5mm}-\hspace{-0.5mm}x_5(v_{1n}\hspace{-0.5mm}-\hspace{-0.5mm}v_{4n})\\-y_5(v_{1n}\hspace{-0.5mm}-\hspace{-0.5mm}v_{3n})\\-z_5(v_{1n}\hspace{-0.5mm}-\hspace{-0.5mm}v_{2n})\end{matrix}$\hspace{-1mm}\\[6mm]
0 & 0 & 0 & 0 & 0 & 0 & 0 & 0 & \hspace{-1mm}$x_6(u_{4n}\hspace{-0.5mm}-\hspace{-0.5mm}u_{6n})$\hspace{-1mm} & \hspace{-1mm}$y_6(u_{3n}\hspace{-0.5mm}-\hspace{-0.5mm}u_{6n})$\hspace{-1mm} & \hspace{-1mm}$z_6(u_{2n}\hspace{-0.5mm}-\hspace{-0.5mm}u_{6n})$\hspace{-1mm} & \hspace{-1mm}$\begin{matrix}(u_{1n}\hspace{-0.5mm}-\hspace{-0.5mm}u_{6n})\hspace{-0.5mm}-\hspace{-0.5mm}x_6(u_{1n}\hspace{-0.5mm}-\hspace{-0.5mm}u_{4n})\\-y_6(u_{1n}\hspace{-0.5mm}-\hspace{-0.5mm}u_{3n})\\-z_6(u_{1n}\hspace{-0.5mm}-\hspace{-0.5mm}u_{2n})\end{matrix}$\hspace{-1mm}\\[6mm]
0 & 0 & 0 & 0 & 0 & 0 & 0 & 0 & \hspace{-1mm}$x_6(v_{4n}\hspace{-0.5mm}-\hspace{-0.5mm}v_{6n})$\hspace{-1mm} & \hspace{-1mm}$y_6(v_{3n}\hspace{-0.5mm}-\hspace{-0.5mm}v_{6n})$\hspace{-1mm} & \hspace{-1mm}$z_6(v_{2n}\hspace{-0.5mm}-\hspace{-0.5mm}v_{6n})$\hspace{-1mm} & \hspace{-1mm}$\begin{matrix}(v_{1n}\hspace{-0.5mm}-\hspace{-0.5mm}v_{6n})\hspace{-0.5mm}-\hspace{-0.5mm}x_6(v_{1n}\hspace{-0.5mm}-\hspace{-0.5mm}v_{4n})\\-y_6(v_{1n}\hspace{-0.5mm}-\hspace{-0.5mm}v_{3n})\\-z_6(v_{1n}\hspace{-0.5mm}-\hspace{-0.5mm}v_{2n})\end{matrix}$\hspace{-1mm}
\end{tabular}\hspace{-2mm}\right]\hspace{-1mm}\begin{bmatrix}p_{1n}\\[1mm]p_{2n}\\[1mm]p_{3n}\\[1mm]p_{4n}\\[1mm]p_{5n}\\[1mm]p_{6n}\\[1mm]p_{7n}\\[1mm]p_{8n}\\[1mm]p_{9n}\\[1mm]p_{10n}\\[1mm]p_{11n}\\[1mm]p_{12n}\end{bmatrix} = {\bf 0}_{12\times 1\, .}}
\end{equation}

From Eq.~\eqref{eq37}, we easily get some simple linear equations of $\{p_{in}\}$ and $\{u_{mn}\}$ such as $p_{4n} = u_{1n}p_{12n},~p_{8n} = v_{1n}p_{12n}$ or $p_{3n} = u_{2n}p_{11n} - (u_{1n}-u_{2n})p_{12n}$. However the interesting point we exploit from Eq.~\eqref{eq37} to get the upper rank given by this theorem is in the following equation
\setcounter{equation}{37}
\begin{equation}\label{eq38}
\widehat{\bf K}\begin{bmatrix}p_{9n} \\ p_{10n} \\ p_{11n} \\ p_{12n}\end{bmatrix} = {\bf 0}_{1\times 4\, ,}
\end{equation}
where
\begin{equation}\label{eq39}
\widehat{\bf K} ~\defeq~
\begin{bmatrix}
x_5(u_{4n}-u_{5n}) & y_5(u_{3n}-u_{5n}) & z_5(u_{2n}-u_{5n}) & \begin{matrix}(u_{1n}-u_{5n})-x_5(u_{1n}-u_{4n})\\-y_5(u_{1n}-u_{3n})\\-z_5(u_{1n}-u_{2n})\end{matrix}\\[7mm]
x_5(v_{4n}-v_{5n}) & y_5(v_{3n}-v_{5n}) & z_5(v_{2n}-v_{5n}) & \begin{matrix}(v_{1n}-v_{5n})-x_5(v_{1n}-v_{4n})\\-y_5(v_{1n}-v_{3n})\\-z_5(v_{1n}-v_{2n})\end{matrix}\\[7mm]
x_6(u_{4n}-u_{6n}) & y_6(u_{3n}-u_{6n}) & z_6(u_{2n}-u_{6n}) & \begin{matrix}(u_{1n}-u_{6n})-x_6(u_{1n}-u_{4n})\\-y_6(u_{1n}\-u_{3n})\\-z_6(u_{1n}-u_{2n})\end{matrix}\\[7mm]
x_6(v_{4n}-v_{6n}) & y_6(v_{3n}-v_{6n}) & z_6(v_{2n}-v_{6n}) & \begin{matrix}(v_{1n}-v_{6n})-x_6(v_{1n}-v_{4n})\\-y_6(v_{1n}-v_{3n})\\-z_6(v_{1n}-v_{2n})\end{matrix}
\end{bmatrix}_{4\times 4\, .}
\end{equation}
It means that the matrix $\widehat{\bf K}$ is not full-rank and its determinant is zero, i.e.,
\begin{equation}\label{eq40}
\text{det}(\widehat{\bf K}) = 0\, .
\end{equation}
Noting that the variables in Eq.~\eqref{eq40} is $x_5,y_5,z_5$, the coordinates of ${\bf X}_5$ and $x_6,y_6,z_6$, the coordinates of ${\bf X}_6$. We do not have any variable in ${\bf P}_n$ in Eq.~\eqref{eq40}.
To compute the determinant of $\widehat{\bf K}$, we notice that
\begin{equation}\label{eq41}
\begin{split}
&\begin{vmatrix}x_5(u_{4n}-u_{5n})&y_5(u_{3n}-u_{5n})\\[2mm]
x_5(v_{4n}-v_{5n})&y_5(v_{3n}-v_{5n})\end{vmatrix}\begin{vmatrix}
z_6(u_{2n}-u_{6n})&\begin{matrix}(u_{1n}-u_{6n})-x_6(u_{1n}-u_{4n})\\-y_6(u_{1n}-u_{3n})\\-z_6(u_{1n}-u_{2n})\end{matrix}\\[9mm]
z_6(v_{2n}-v_{6n})&\begin{matrix}(v_{1n}-v_{6n})-x_6(v_{1n}-v_{4n})\\-y_6(v_{1n}-v_{3n})\\-z_6(v_{1n}-v_{2n})\end{matrix}
\end{vmatrix}\\[2mm]
&= x_5y_5z_6\xi^{\text{\tiny(n)}}_{34,5}\big[\xi^{\text{\tiny(n)}}_{12,6}-x_6\xi^{\text{\tiny(n)}}_{12,46}-y_6\xi^{\text{\tiny(n)}}_{12,36}-z_6\xi^{\text{\tiny(n)}}_{12,26}\big]
\end{split}
\end{equation}
where $\xi^{\text{\tiny(n)}}_{34,5}$ and $\xi^{\text{\tiny(n)}}_{12,6}$ are defined by Eq.~\eqref{eq32}, and others are given by
\begin{equation}\label{eq42}
\xi^{\text{\tiny(n)}}_{i_1i_2,j_1j_2}\,\defeq\,\begin{vmatrix}u_{i_1n}-u_{j_1n} & u_{i_2n}-u_{j_2n}\\[1mm]v_{i_1n}-v_{j_1n} & v_{i_2n}-v_{j_2n}\end{vmatrix}.
\end{equation}
Now the determinant of $\widehat{\bf K}$ is computed via its $2\times 2$ blocks as follows
\begin{equation}\label{eq43}
\begin{split}
\text{det}(\widehat{\bf K}) &= x_5y_5z_6\xi^{\text{\tiny(n)}}_{34,5}\big[\xi^{\text{\tiny(n)}}_{12,6}-x_6\xi^{\text{\tiny(n)}}_{12,46}-y_6\xi^{\text{\tiny(n)}}_{12,36}-z_6\xi^{\text{\tiny(n)}}_{12,26}\big]\\[0.5mm]
&~~~ - x_5y_6z_5\xi^{\text{\tiny(n)}}_{24,5}\big[\xi^{\text{\tiny(n)}}_{13,6}-x_6\xi^{\text{\tiny(n)}}_{13,46}-y_6\xi^{\text{\tiny(n)}}_{13,36}-z_6\xi^{\text{\tiny(n)}}_{13,26}\big]\\[0.5mm]
&~~~ + x_6y_5z_5\xi^{\text{\tiny(n)}}_{23,5}\big[\xi^{\text{\tiny(n)}}_{14,6}-x_6\xi^{\text{\tiny(n)}}_{14,46}-y_6\xi^{\text{\tiny(n)}}_{14,36}-z_6\xi^{\text{\tiny(n)}}_{14,26}\big]\\[0.5mm]
&~~~ + x_5y_6z_6\xi^{\text{\tiny(n)}}_{23,6}\big[\xi^{\text{\tiny(n)}}_{14,5}-x_5\xi^{\text{\tiny(n)}}_{14,45}-y_5\xi^{\text{\tiny(n)}}_{14,35}-z_5\xi^{\text{\tiny(n)}}_{14,25}\big]\\[0.5mm]
&~~~ - x_6y_5z_6\xi^{\text{\tiny(n)}}_{24,6}\big[\xi^{\text{\tiny(n)}}_{13,5}-x_5\xi^{\text{\tiny(n)}}_{13,45}-y_5\xi^{\text{\tiny(n)}}_{13,35}-z_5\xi^{\text{\tiny(n)}}_{13,25}\big]\\[0.5mm]
&~~~ + x_6y_6z_5\xi^{\text{\tiny(n)}}_{34,6}\big[\xi^{\text{\tiny(n)}}_{12,5}-x_5\xi^{\text{\tiny(n)}}_{12,45}-y_5\xi^{\text{\tiny(n)}}_{12,35}-z_5\xi^{\text{\tiny(n)}}_{12,25}\big]
\end{split}
\end{equation}
We present $\text{det}(\widehat{\bf K})$ as a polynomials in which its monomials are
\begin{equation}\label{eq44}
\begin{split}
&x_5y_5z_6,\, x_5y_6z_5,\, x_6y_5z_5,\, x_6y_6z_5,\, x_6y_5z_6,\, x_5y_6z_6,\\
&x_5^2y_6z_6,\, x_6y_5^2z_6,\, x_6y_6z_5^2,\, x_6^2y_5z_5,\, x_5y_6^2z_5,\, x_5y_5z_6^2\\
&x_5y_5x_6z_6,\, x_5y_5y_6z_6,\, x_5z_5x_6y_6,\, x_5z_5y_6z_6,\, y_5z_5x_6y_6,\, y_5z_5x_6z_6\, .
\end{split}
\end{equation}
This polynomial is
\begin{equation}\label{eq45}
\text{det}(\widehat{\bf K}) = b_0x_5y_5z_6 + b_1x_5y_6z_5 + \cdots + b_{17}y_5z_5x_6z_6\, ,
\end{equation}
where
\begin{equation}\label{eq46}
\begin{split}
b_0&= - b_{11} = \xi^{\text{\tiny(n)}}_{34,5}\xi^{\text{\tiny(n)}}_{12,6},~\hspace{3.75mm}b_1 \hspace{0.25mm}= - b_{10} =  \xi^{\text{\tiny(n)}}_{42,5}\xi^{\text{\tiny(n)}}_{13,6},~\hspace{3.75mm}b_2 = -b_9\hspace{1.5mm} = \xi^{\text{\tiny(n)}}_{23,5}\xi^{\text{\tiny(n)}}_{14,6}\, ,\\[1mm]
b_3&= - b_8 \hspace{1.25mm}= \xi^{\text{\tiny(n)}}_{34,6}\xi^{\text{\tiny(n)}}_{12,5},~\hspace{3.75mm}b_4= -b_7\hspace{1.75mm} = \xi^{\text{\tiny(n)}}_{42,6}\xi^{\text{\tiny(n)}}_{13,5},~~~~\,b_5 = -b_6 \hspace{1mm}= \xi^{\text{\tiny(n)}}_{23,6}\xi^{\text{\tiny(n)}}_{14,5}\, ,\\[1mm]
b_{12}&= \xi^{\text{\tiny(n)}}_{24,6}\xi^{\text{\tiny(n)}}_{13,45} - \xi^{\text{\tiny(n)}}_{34,5}\xi^{\text{\tiny(n)}}_{12,46},~b_{13}= \xi^{\text{\tiny(n)}}_{32,6}\xi^{\text{\tiny(n)}}_{14,35} - \xi^{\text{\tiny(n)}}_{34,5}\xi^{\text{\tiny(n)}}_{12,36},~b_{14}= \xi^{\text{\tiny(n)}}_{24,5}\xi^{\text{\tiny(n)}}_{13,46} - \xi^{\text{\tiny(n)}}_{34,6}\xi^{\text{\tiny(n)}}_{12,45}\, ,\\[1mm]
b_{15}&= \xi^{\text{\tiny(n)}}_{24,5}\xi^{\text{\tiny(n)}}_{13,26} - \xi^{\text{\tiny(n)}}_{23,6}\xi^{\text{\tiny(n)}}_{14,25},~b_{16}= \xi^{\text{\tiny(n)}}_{32,5}\xi^{\text{\tiny(n)}}_{14,36} - \xi^{\text{\tiny(n)}}_{34,6}\xi^{\text{\tiny(n)}}_{12,35},~ b_{17}= \xi^{\text{\tiny(n)}}_{24,6}\xi^{\text{\tiny(n)}}_{13,25} -\xi^{\text{\tiny(n)}}_{23,5}\xi^{\text{\tiny(n)}}_{14,26}\, .
\end{split}
\end{equation}
We use the following explanation
\begin{equation}\label{eq47}
\begin{split}
\xi^{\text{\tiny(n)}}_{i_1j_1,i_2j_2}-\xi^{\text{\tiny(n)}}_{i_1j_1,i_3j_2} &= (u_{i_1n}-u_{i_2n})(v_{j_1n}-v_{j_2n}) - (v_{i_1n}-v_{i_2n})(u_{j_1n}-u_{j_2n})\\
&~~ - (u_{i_1n}-u_{i_3n})(v_{j_1n}-v_{j_2n}) + (v_{i_1n}-v_{i_3n})(u_{j_1n}-u_{j_2n})\\
&= (u_{i_3n}-u_{i_2n})(v_{j_1n}-v_{j_2n}) - (v_{i_3n}-v_{i_2n})(u_{j_1n}-u_{j_2n})\\
&= \xi^{\text{\tiny(n)}}_{i_3j_1,i_2j_2}\, ,
\end{split}
\end{equation}
to obtain
\begin{equation}\label{eq48}
\begin{split}
b_4 + b_{12} &= \xi^{\text{\tiny(n)}}_{24,6}\big(\xi^{\text{\tiny(n)}}_{13,45}-\xi^{\text{\tiny(n)}}_{13,5}\big) - \xi^{\text{\tiny(n)}}_{34,5}\xi^{\text{\tiny(n)}}_{12,46} = \xi^{\text{\tiny(n)}}_{24,6}\xi^{\text{\tiny(n)}}_{34,5} - \xi^{\text{\tiny(n)}}_{34,5}\xi^{\text{\tiny(n)}}_{12,46} = -\xi^{\text{\tiny(n)}}_{34,5}\xi^{\text{\tiny(n)}}_{12,6}\\
&= -b_0\, ,
\end{split}
\end{equation}
and similar others
\begin{equation}\label{eq49}
\begin{split}
&b_{13} + b_5 + b_0 = 0,~~ b_{14} + b_3 + b_1 = 0,~~ b_{15} + b_5 + b_1 = 0,\\
&b_{16} + b_3 + b_2 = 0,~~ b_{17} + b_4 + b_2 = 0,~~ b_{16} + b_{15} + b_{12} = 0,~~ b_{17} + b_{14} + b_{13} = 0\, .\\[-3mm]
\end{split}
\end{equation}
Combining Eq.~\eqref{eq45}, Eq.~\eqref{eq46}, Eq.~\eqref{eq48} and Eq.~\eqref{eq49}, the determinant of $\widehat{\bf K}$ can be presented as a polynomial with five monomials as follows.
\begin{equation}\label{eq50}
\begin{split}
&\text{det}(\widehat{\bf K}) =
{\begin{bmatrix}b_0\\[1mm]b_1\\[1mm]b_2\\[1mm]b_3\\[1mm]b_4\end{bmatrix}^T\hspace{-1mm}\begin{bmatrix}
x_5z_6\big[y_5(1 - x_6 - y_6 - z_6) - y_6(1 - x_5 - y_5 - z_5)\big]\\[1mm]
x_5y_6\big[z_5(1 - x_6 - y_6 - z_6) - z_6(1 - x_5 - y_5 - z_5)\big]\\[1mm]
x_6y_5z_5(1 - x_6 - y_6 - z_6) - x_5y_6z_6(1 - x_5 - y_5 - z_5)\\[1mm]
(x_6z_5-x_5z_6)y_6(1 - x_5 - y_5 - z_5)\\[1mm]
(x_6y_5-x_5y_6)z_6(1 - x_5 - y_5 - z_5)
\end{bmatrix}}
\end{split}
\end{equation}
Since $\text{det}(\widehat{\bf K}) = 0$, the two column-vectors in Eq.~\eqref{eq50} are orthogonal. Note that $[b_0\,,\, b_1\,,\, b_2\,,\, b_3\,,\, b_4]$ is the $n$th row of the matrix ${\bf \Lambda}$ in Theorem 2. Hence all the rows of this matrix are orthogonal to the non-zero vector constructed by the coordinates of ${\bf X}_5$ and ${\bf X}_6$ given by Eq.~\eqref{eq50}. In another word, the $N\times 5$ matrix ${\bf \Lambda}$ has the rank at most 4.
\end{proof}

Similarly developing the rank property to the self-estimation as we do with Theorem~\ref{Thm:RanktwoView} and Proposition~\ref{Prop:SelftwoView} in the last sub-section, we derive Proposition~\ref{Prop:SelffiveView} as an application of the rank property proposed by Theorem~\ref{Thm:RankfiveView} on the self-estimation.

\begin{prop}\label{Prop:SelffiveView}
Assuming that $N\geq 6$, the functions of $u_{61}$ and $v_{61}$ from the remain coordinates are given by
\begin{equation}\label{eq51}
\begin{bmatrix}u_{61}\\v_{61}\end{bmatrix} = -\begin{bmatrix}\hat{\alpha}_{1|2,3,4,5} & \hat{\beta}_{1|2,3,4,5}\\
\vdots & \vdots\\\hat{\alpha}_{1|j_1j_2j_3j_4} & \hat{\beta}_{1|j_1j_2j_3j_4}\\\vdots & \vdots\\\hat{\alpha}_{1|N-3,\ldots,N} & \hat{\beta}_{1|N-3,\ldots,N}\end{bmatrix}^{\dagger}\begin{bmatrix}\hat{\gamma}_{1|2,3,4,5}\\\vdots\\\hat{\gamma}_{1|j_1j_2j_3j_4}\\\vdots\\\hat{\gamma}_{1|N-3,\ldots,N}\end{bmatrix}
\end{equation}
for all $2\leq j_1 < j_2 < j_3 < j_4 \leq N$, where
\begin{equation}\label{eq52}
\begin{split}
    \hat{\alpha}_{1|j_1j_2j_3j_4} &= \begin{vmatrix}\begin{matrix}\xi^{\text{\tiny(1)}}_{34,5}(v_{11}-v_{21})\\[1mm]\xi^{\text{\tiny(1)}}_{42,5}(v_{11}-v_{31})\\[1mm]\xi^{\text{\tiny(1)}}_{23,5}(v_{11}-v_{41})\\[1mm]\xi^{\text{\tiny(1)}}_{12,5}(v_{31}-v_{41})\\[1mm]\xi^{\text{\tiny(1)}}_{13,5}(v_{41}-v_{21})\end{matrix} ~&~ {\bf A}_{j_1j_2j_3j_4}\end{vmatrix}_{5\times 5}\quad,\\[3mm]
    \hat{\beta}_{1|j_1j_2j_3j_4} &= \begin{vmatrix}\begin{matrix}\xi^{\text{\tiny(1)}}_{34,5}(u_{21}-u_{11})\\[1mm]\xi^{\text{\tiny(1)}}_{42,5}(u_{31}-u_{11})\\[1mm]\xi^{\text{\tiny(1)}}_{23,5}(u_{41}-u_{11})\\[1mm]\xi^{\text{\tiny(1)}}_{12,5}(u_{41}-u_{31})\\[1mm]\xi^{\text{\tiny(1)}}_{13,5}(u_{21}-u_{41})\end{matrix} ~&~ {\bf A}_{j_1j_2j_3j_4}\end{vmatrix}_{5\times 5}\quad,\\[3mm]
    \hat{\gamma}_{1|j_1j_2j_3j_4} &= \begin{vmatrix}\begin{matrix}\xi^{\text{\tiny(1)}}_{34,5}(u_{11}v_{21}-u_{21}v_{11})\\[1mm]\xi^{\text{\tiny(1)}}_{42,5}(u_{11}v_{31}-u_{31}v_{11})\\[1mm]\xi^{\text{\tiny(1)}}_{23,5}(u_{11}v_{41}-u_{41}v_{11})\\[1mm]\xi^{\text{\tiny(1)}}_{12,5}(u_{31}v_{41}-u_{41}v_{31})\\[1mm]\xi^{\text{\tiny(1)}}_{13,5}(u_{41}v_{21}-u_{21}v_{41})\end{matrix} ~&~ {\bf A}_{j_1j_2j_3j_4}\end{vmatrix}_{5\times 5}\quad,
\end{split}
\end{equation}
and
\begin{equation}\label{eq53}
{\bf A}_{j_1j_2j_3j_4} \hspace{-0.5mm}=\hspace{-0.5mm} {\begin{bmatrix}
\xi^{\text{\tiny($j_1$)}}_{34,5}\xi^{\text{\tiny($j_1$)}}_{12,6} & \xi^{\text{\tiny($j_2$)}}_{34,5}\xi^{\text{\tiny($j_2$)}}_{12,6} & \xi^{\text{\tiny($j_3$)}}_{34,5}\xi^{\text{\tiny($j_3$)}}_{12,6} &
\xi^{\text{\tiny($j_4$)}}_{34,5}\xi^{\text{\tiny($j_4$)}}_{12,6}\\[3mm]
\xi^{\text{\tiny($j_1$)}}_{42,5}\xi^{\text{\tiny($j_1$)}}_{13,6} & \xi^{\text{\tiny($j_2$)}}_{42,5}\xi^{\text{\tiny($j_2$)}}_{13,6} & \xi^{\text{\tiny($j_3$)}}_{42,5}\xi^{\text{\tiny($j_3$)}}_{13,6} &
\xi^{\text{\tiny($j_4$)}}_{42,5}\xi^{\text{\tiny($j_4$)}}_{13,6}\\[3mm]
\xi^{\text{\tiny($j_1$)}}_{23,5}\xi^{\text{\tiny($j_1$)}}_{14,6} & \xi^{\text{\tiny($j_2$)}}_{23,5}\xi^{\text{\tiny($j_2$)}}_{14,6} & \xi^{\text{\tiny($j_3$)}}_{23,5}\xi^{\text{\tiny($j_3$)}}_{14,6} &
\xi^{\text{\tiny($j_4$)}}_{23,5}\xi^{\text{\tiny($j_4$)}}_{14,6}\\[3mm]
\xi^{\text{\tiny($j_1$)}}_{12,5}\xi^{\text{\tiny($j_1$)}}_{34,6} & \xi^{\text{\tiny($j_2$)}}_{12,5}\xi^{\text{\tiny($j_2$)}}_{34,6} & \xi^{\text{\tiny($j_3$)}}_{12,5}\xi^{\text{\tiny($j_3$)}}_{34,6} &
\xi^{\text{\tiny($j_4$)}}_{12,5}\xi^{\text{\tiny($j_4$)}}_{34,6}\\[2.5mm] \xi^{\text{\tiny($j_1$)}}_{13,5}\xi^{\text{\tiny($j_1$)}}_{42,6} &
\xi^{\text{\tiny($j_2$)}}_{13,5}\xi^{\text{\tiny($j_2$)}}_{42,6} &
\xi^{\text{\tiny($j_3$)}}_{13,5}\xi^{\text{\tiny($j_3$)}}_{42,6} &
\xi^{\text{\tiny($j_4$)}}_{13,5}\xi^{\text{\tiny($j_4$)}}_{42,6}
\end{bmatrix}_{5\times 4\, .}}
\end{equation}
\end{prop}
\begin{proof}
For any five indices $1$ and $2 \leq j_1 < j_2 < j_3 < j_4 \leq N$, we use five rows of ${\bf \Lambda}$ in Eq.~\eqref{eq31} corresponding with these five indices to get the $5\times5$ sub-matrix ${\bf \Lambda}_{1j_1j_2j_3j_4}$. The rank property of ${\bf \Lambda}$ yields
\begin{equation}\label{eq54}
    \text{det}\big({\bf \Lambda}_{1j_1j_2j_3j_4}\big) = 0\, ,
\end{equation}
or
\begin{equation}\label{eq55}
\begin{bmatrix}
\xi^{\text{\tiny(1)}}_{34,5}\xi^{\text{\tiny(1)}}_{12,6}\\[3mm]
-\xi^{\text{\tiny(1)}}_{42,5}\xi^{\text{\tiny(1)}}_{13,6}\\[3mm]
\xi^{\text{\tiny(1)}}_{23,5}\xi^{\text{\tiny(1)}}_{14,6}\\[3mm]
-\xi^{\text{\tiny(1)}}_{12,5}\xi^{\text{\tiny(1)}}_{34,6}\\[3mm]
\xi^{\text{\tiny(1)}}_{13,5}\xi^{\text{\tiny(1)}}_{42,6}
\end{bmatrix}^T\begin{bmatrix}
\big|{\bf \Lambda}_{1j_1j_2j_3j_4}^1\big|\\[3mm]
\big|{\bf \Lambda}_{1j_1j_2j_3j_4}^2\big|\\[3mm]
\big|{\bf \Lambda}_{1j_1j_2j_3j_4}^3\big|\\[3mm]
\big|{\bf \Lambda}_{1j_1j_2j_3j_4}^4\big|\\[3mm]
\big|{\bf \Lambda}_{1j_1j_2j_3j_4}^5\big|
\end{bmatrix}
~= 0\, .
\end{equation}
Since
\begin{equation}\label{eq56}
\xi^{\text{\tiny(1)}}_{i_1i_2,6} = (u_{i_11}v_{i_21} - u_{i_21}v_{i_11}) + (v_{i_11}-v_{i_21})u_{61} + (u_{i_21}-u_{i_11})v_{61}
\end{equation}
Eq.~\eqref{eq55} means
\begin{equation}\label{eq57}
\hat{\alpha}_{1|j_1j_2j_3j_4}u_{61} + \hat{\beta}_{1|j_1j_2j_3j_4}v_{61} + \hat{\gamma}_{1|j_1j_2j_3j_4} = 0\, ,
\end{equation}
where $\hat{\alpha}_{1|j_1j_2j_3j_4}, \hat{\beta}_{1|j_1j_2j_3j_4}$ and $\hat{\gamma}_{1|j_1j_2j_3j_4}$ are given by Eq.~\eqref{eq52} and Eq.~\eqref{eq53}. Let four indices $(j_1,j_2,j_3,j_4)$ run from the first $(2,3,4,5)$ to the last $(N-3,N-2,N-1,N)$, we get a similar formula of Eq.~\eqref{eq29} that
\begin{equation}\label{eq58}
\begin{bmatrix}\hat{\alpha}_{1|2,3,4,5} & \hat{\beta}_{1|2,3,4,5}\\\vdots & \vdots\\\hat{\alpha}_{1|j_1j_2j_3j_4} & \hat{\beta}_{1|j_1j_2j_3j_4}\\\vdots & \vdots\\\hat{\alpha}_{1|N-3,\ldots,N} & \hat{\beta}_{1|N-3,\ldots,N}\end{bmatrix}\begin{bmatrix}u_{61}\\v_{61}\end{bmatrix} = -\begin{bmatrix}\hat{\gamma}_{1|2,3,4,5}\\\vdots\\\hat{\gamma}_{1|j_1j_2j_3j_4}\\\vdots\\\hat{\gamma}_{1|N-3,\ldots,N}\end{bmatrix}\, .
\end{equation}
When $N \geq 6$, the matrix in the left side of Eq.~\eqref{eq58} with the size ${N-1 \choose 4}\times 2$ is always full rank. So Eq.~\eqref{eq58} and Eq.~\eqref{eq51} are equivalent. Hence, Proposition~\ref{Prop:SelffiveView} is correct.
\end{proof}

\vspace{0.2in}

\section{Correspondences Refinement}\label{sec:refinement}

In this section, we study how to apply Theorem~\ref{Thm:RanktwoView}, Theorem~\ref{Thm:RankfiveView}, Proposition~\ref{Prop:SelftwoView} and Proposition~\ref{Prop:SelffiveView} on refining correspondences. More precisely, we study three applications that (i) the correspondences refinement based on Theorem~\ref{Thm:RanktwoView}, (ii) the outliers recognition based on Theorem~\ref{Thm:RanktwoView}, and (iii) the self correspondences estimation based on Proposition~\ref{Prop:SelftwoView} and Proposition~\ref{Prop:SelffiveView}. Combining these three applications, we finally propose an algorithm, named \emph{Main Algorithm}, to derive better correspondences from their observations.

\newpage

\subsection{Correspondences Refinement}

The $M\times 9$ matrix ${\bf \Gamma}$ given by Eq.~\eqref{eq14} in Theorem~\ref{Thm:RanktwoView} and the $N\times 5$ matrix ${\bf \Lambda}$ given by Eq.~\eqref{eq31} in Theorem~\ref{Thm:RankfiveView} for the perfect correspondences (the ground-truths) without noise are rank deficient. With a set of corrupted noisy correspondence observations, the two matrices will be full rank. The utilization of Theorem~\ref{Thm:RanktwoView} and Theorem~\ref{Thm:RankfiveView} for the application of \emph{correspondences refinement} is to force the rank deficient of ${\bf \Gamma}$ and ${\bf \Lambda}$. It can be achieved through \emph{singular value decomposition} (SVD) and the correspondences retrieved from the approximations will give the refined correspondences.

In the case for Theorem~\ref{Thm:RanktwoView}, for any correspondences from two images and at least nine key-points we first form the matrix ${\bf \Gamma}$ and apply the SVD. By keeping the largest eight singular values and setting the last to zero, an improved ${\bf \Gamma}$ can be obtained. The second, third, forth and fifth columns of this improvement will give the enhanced correspondences for the associated two images. Notice that the rank deficient property of ${\bf \Gamma}$ is still correct when these correspondences are transformed by any translation, rotation or reflection. Thus, we fix these ambiguities by finding an optimal translation, rotation and reflection to transform the enhanced correspondences close to the observed correspondences. We repeat the same processing for all two images and at least nine key-points, collect all enhanced correspondences, take their medians corresponding with each image point, and iterate a few times.
This algorithm is called the \emph{Correspondences refinement} and presented carefully by Algorithm~\ref{Alg:Refinement}.

For Theorem~\ref{Thm:RankfiveView}, we easily obtain the improved ${\bf \Lambda}$ by applying SVD, keeping the largest four singular values and setting the last to zero as we do with ${\bf \Gamma}$. However, it seems very difficult to extract the enhanced correspondences from the improvement of ${\bf \Lambda}$. In this work, we have not done with this problem, and it will be our the future work.

\begin{algorithm}[t]
\caption{Correspondences Refinement}
\label{Alg:Refinement}
{
{\bf Input:} $\big\{\hat{\bf x}^{\text{\tiny(1)}}_m \,\leftrightarrow\,\hat{\bf x}^{\text{\tiny(2)}}_m \,\leftrightarrow\, \cdots \,\leftrightarrow\,\hat{\bf x}^{\text{\tiny(N)}}_m\big\}^M_{m=1}$~:~ Correspondences.\\[-2mm]

{\bf Implementation:}\\[-2mm]

\textbf{\emph{Repeat several times the following steps:}}\\[-2mm]

{\bf 1.} $\mathcal{X}_{mn} \longleftarrow \emptyset$~:~ A set to store all candidates for refining $\hat{\bf x}_m^{\text{\tiny(n)}}$.\\[-2mm]

{\bf 2.} $\bar{\bf x}_m^{\text{\tiny(n)}} ~~\longleftarrow~ \hat{\bf x}_m^{\text{\tiny (n)}}$~ for all $n,m$.\\[-2mm]

{\bf 3.} For all $1 \leq n_1 < n_2 \leq N$\\[-2mm]

\hspace{5mm}{\bf 3.1.} ${\bf \Gamma} ~\overset{\eqref{eq13}\text{ in Theorem~\ref{Thm:RanktwoView}}}{\longleftarrow}~ \Bigg\{\begin{bmatrix}u_m\\[1mm]v_m\end{bmatrix} \defeq \bar{\bf x}_{m}^{\text{\tiny($n_1$)}}~\text{ and } ~ \begin{bmatrix}u'_m\\[1mm]v'_m\end{bmatrix} \defeq \bar{\bf x}_{m}^{\text{\tiny($n_2$)}}\Bigg\}^M_{m=1}$\\

\hspace{5mm}{\bf 3.2.} Apply SVD and use the largest 8 singular values only to construct $\widehat{\bf \Gamma}$ as a refinement of ${\bf \Gamma}$.\\[-2mm]

\hspace{5mm}{\bf 3.3.} Find the optimal candidates $\big\{\breve{\bf x}_{1}^{\text{\tiny($n_1$)}},\breve{\bf x}_{1}^{\text{\tiny($n_2$)}},\ldots,\breve{\bf x}_{M}^{\text{\tiny($n_1$)}},\breve{\bf x}_{M}^{\text{\tiny($n_2$)}}\big\}$~:\\[-1mm]

\hspace{10mm}{\bf a.} $\breve{\bf x}_{m}^{\text{\tiny($n_1$)}} \leftarrow \begin{bmatrix}\widehat{\bf \Gamma}_{m,2}\\[0.5mm]\widehat{\bf \Gamma}_{m,3}\end{bmatrix}$ ~and~ $\breve{\bf x}_{m}^{\text{\tiny($n_2$)}} \leftarrow \begin{bmatrix}\widehat{\bf \Gamma}_{m,4}\\[0.5mm]\widehat{\bf \Gamma}_{m,5}\end{bmatrix}$\, for~ $m = 1,\ldots,M,$ where $\widehat{\bf \Gamma}_{i,j}$ is the $(i,j)$th element of $\widehat{\bf \Gamma}$.\\[1mm]

\hspace{10mm}{\bf b.} To avoid the ambiguities of translations, rotations and reflections we use the optimal translation, rotation and reflection to change $\big\{\breve{\bf x}_{1}^{\text{\tiny($n_1$)}},\breve{\bf x}_{1}^{\text{\tiny($n_2$)}},\ldots,\breve{\bf x}_{M}^{\text{\tiny($n_1$)}},\breve{\bf x}_{M}^{\text{\tiny($n_2$)}}\big\}$ such that\\[-4mm]

\[\sum^2_{j=1}\sum^M_{m=1}\big\|\breve{\bf x}_{m}^{\text{\tiny($n_j$)}} - \hat{\bf x}_{m}^{\text{\tiny($n_j$)}}\big\|^2_2 \quad\longrightarrow\quad \text{minimum\, .}\]

\hspace{10mm}{\bf c.} Store $\breve{\bf x}_{m}^{\text{\tiny($n_1$)}}$ to $\mathcal{X}_{mn_1}$ ~and~ $\breve{\bf x}_{m}^{\text{\tiny($n_2$)}}$ to $\mathcal{X}_{mn_2}$\, .\\[-2mm]

{\bf 4.} For all $n,m$, obtain the refinement $\bar{\bf x}_m^{\text{\tiny(n)}}$ of $\hat{\bf x}_m^{\text{\tiny(n)}}$, by taking the \emph{median} of all elements in $\mathcal{X}_{mn}$.\\[-1mm]

{\bf Output:} $\big\{\bar{\bf x}^{\text{\tiny(1)}}_m \,\leftrightarrow\,\bar{\bf x}^{\text{\tiny(2)}}_m \,\leftrightarrow\,\cdots \,\leftrightarrow\,\bar{\bf x}^{\text{\tiny(N)}}_m\big\}^M_{m=1}$~: The refinements.
}
\end{algorithm}

\subsection{Outliers Recognition}

The Algorithm~\ref{Alg:Refinement}, \emph{Correspondences Refinement}, gives us the refined correspondences from the observed correspondences. When the noises are absent, i.e., the observations are the ground-truths, a difference between all the refinements and all the observations will be zero. We expect this difference as a parameter to evaluate how large the noises are in the observations. The small difference confirms the very good observed correspondences and vice versa, the large difference yields the bad ones. In addition, if the difference from all correspondences is small but a difference from one image point and its refinement is large, we will recognize this image point an outlier. Here we emphasize that the concept \emph{outlier} is applied for an image point, not for a correspondence. Thus, an outlier is an $2\times 1$ vector in which the first coordinate (an integer number between one and $N$) points to the image and the second one (an integer number between one and $M$) indicates the point.

Following the above explanation, we propose Algorithm~\ref{Alg:Recognition}, named \emph{Outliers recognition}, to detect outliers from the observations. Given a threshold $\theta$, we detect an outlier of an image point by evaluating an Euclidean distance between this image point and its refined image point we get from the \emph{Correspondences refinement} algorithm. If this distance is larger than the threshold $\theta$, this image point is an outlier. Note that coordinates of image points are integer numbers that are bounded by zero and two sizes of images. In this work, we expect to obtain correspondences in which their errors are smaller than 50 pixels, thus we often choose $\theta$ as an integer number between 20 and 70.

\begin{algorithm}[t]
\caption{Outliers Recognition}
\label{Alg:Recognition}
{
{\bf Input:} $\big\{\hat{\bf x}^{\text{\tiny(1)}}_m \,\leftrightarrow\,\hat{\bf x}^{\text{\tiny(2)}}_m \,\leftrightarrow\,\cdots \,\leftrightarrow\,\hat{\bf x}^{\text{\tiny(N)}}_m\big\}^M_{m=1}$~:~ Correspondences\\[-2mm]

~\hspace{11mm} $\theta$~:~ A threshold to recognize outliers.\\[-2mm]

{\bf Implementation:}\\[-2mm]

{\bf 1.} $\mathcal{O} \longleftarrow \emptyset$~:~ A set to store indices of outliers.\\[-2mm]

{\bf 2.} \textbf{Until} no new outliers are found in below steps:\\[-1mm]

\hspace{5mm}{\bf 2.1} $\big\{\bar{\bf x}_m^{\text{\tiny(1)}}\leftrightarrow\cdots\leftrightarrow\bar{\bf x}_m^{\text{\tiny(N)}}\big\}\overset{\text{Correspondences Refinement}}{\longleftarrow} \big\{\hat{\bf x}_m^{\text{\tiny(1)}}\leftrightarrow\cdots\leftrightarrow\hat{\bf x}_m^{\text{\tiny(N)}}\big\}$.\\[-1mm]

\hspace{5mm}{\bf 2.2} For all $1 \leq n \leq N$ and $1\leq m \leq M$,~~if $\|\hat{\bf x}_m^{\text{\tiny(n)}} - \bar{\bf x}_m^{\text{\tiny(n)}}\|_2 \geq \theta$\\[-1mm]

\hspace{10mm}{\bf a.} A new outlier $\hat{\bf x}_m^{\text{\tiny(n)}}$ is found: Adding the index $[n,m]^T$ to $\mathcal{O}$.\\[-2mm]

\hspace{10mm}{\bf b.} Remove $\hat{\bf x}_m^{\text{\tiny(n)}}$ from $\big\{\hat{\bf x}_m^{\text{\tiny(1)}}\leftrightarrow\cdots\leftrightarrow\hat{\bf x}_m^{\text{\tiny(N)}}\big\}$.\\[-2mm]

\hspace{10mm}{\bf c.} Repeat steps {\bf 2.1} and {\bf 2.2}.\\[-2mm]

{\bf Output:} $\mathcal{O}$~: The indices of outliers.
}
\end{algorithm}

\subsection{Self Correspondences Estimation}

Given the two-correspondences of at least ten key-points $\{\hat{\bf x}_{m_i}^{\text{\tiny(n)}}\leftrightarrow \hat{\bf x}_{m_i}^{\text{\tiny(n')}}\}^K_{i=1}$ $(K\geq 10)$, Proposition~\ref{Prop:SelftwoView} proposes the closed-form solution for one image point, for example $\hat{\bf x}_{m_j}^{\text{\tiny(n)}}$, from others. This result tells us to think about the situation that if the image point $\hat{\bf x}_{m_j}^{\text{\tiny(n)}}$ cannot observed or it is an outlier, while others from the at least ten two-correspondences are good estimations, we can recover the value for $\hat{\bf x}_{m_j}^{\text{\tiny(n)}}$ based on the closed-form given by Proposition~\ref{Prop:SelftwoView}. Similarly, given the `at least six'-correspondences of six key-points, Proposition~\ref{Prop:SelffiveView} gives a closed-form solution for any image point from the others. Thus, if these `at least six'-correspondences contains only one outlier, this unobservable image point or outlier can be recovered.

We notice that given an outlier, there are a lot of `possible' two-correspondences of at least ten key-points and `at least six'-correspondences of six key-points containing this outlier. We use the word `possible' to emphasize that these two- and `at least six'-correspondences has only one outlier. For each `possible' correspondences, we get one candidate for this outlier. Thus, there are lots of candidates for each outlier. From these candidates, we can take their average or their median to recover the value of the outlier. In this work, we use the median as Algorithm~\ref{Alg:Self_Estimation} \emph{Self estimation}.

\begin{algorithm}[t]
\caption{Self Estimation}
\label{Alg:Self_Estimation}
{
{\bf Input:} $\big\{\hat{\bf x}^{\text{\tiny(1)}}_m \,\leftrightarrow\,\hat{\bf x}^{\text{\tiny(2)}}_m \,\leftrightarrow\,\cdots \,\leftrightarrow\,\hat{\bf x}^{\text{\tiny(N)}}_m\big\}^M_{m=1}$~:~ Correspondences\\[-2mm]

~\hspace{11mm} $\mathcal{O}$~:~ The indices of outliers.\\[-2mm]

{\bf Implementation:}\\[-2mm]

For each outlier labeled by the index $[n,m]^T$ in $\mathcal{O}$, we re-estimate a value for $\hat{\bf x}_m^{\text{\tiny(n)}}$ by following:\\[-2mm]

{\bf 1.} $\mathcal{C}_{mn} \longleftarrow \emptyset$~:~ A set to store all candidates of the outlier $\hat{\bf x}_m^{\text{\tiny(n)}}$.\\[-1mm]

{\bf 2.} \textbf{\emph{Use Theorem 1:}}~ For all $1 \leq n' \leq N$ such that $n'\neq n$ and for all $1\leq m_1 < m_2 < m_3 < m_4 < m_5 < m_6 < m_7 < m_8 \leq M$ such that $m\notin\{m_1,m_2,m_3,m_4,m_5,m_6,m_7,m_8\}$,\\[-1mm]

\hspace{5mm}{\bf 2.1.} $\big\{\alpha_{m|m_1\ldots,m_8}, \beta_{m|m_1,\ldots,m_8}, \gamma_{m|m_1,\ldots,m_8}\big\} \longleftarrow \eqref{eq22}$ when\\[-2mm]

\hspace{13mm}$M$-indices $\{1,2,\ldots,9\}\longleftarrow\{m,m_1,\ldots,m_8\}$ and $\hat{\bf x}_{m_i}^{\text{\tiny(n)}} = [u_{m_i},v_{m_i}]^T$~,\quad $\hat{\bf x}_{m_i}^{\text{\tiny(n')}} = [u'_{m_i},v'_{m_i}]^T$. \\[-1mm]

\hspace{5mm}{\bf 2.2.} $\breve{\bf x}_m^{\text{\tiny(n)}} ~\overset{\eqref{eq23}}{\longleftarrow}~ \big\{\alpha_{m|m_1\ldots,m_8}\,,\, \beta_{m|m_1,\ldots,m_8}\,,\,\gamma_{m|m_1,\ldots,m_8}\big\}$\\[-1mm]

\hspace{5mm}{\bf 2.3.} add the candidate $\breve{\bf x}_m^{\text{\tiny(n)}}$ to $\mathcal{C}_{mn}$.\\[-1mm]

{\bf 3.} \textbf{\emph{Use Theorem 2:}} For all $1\leq m_1 < m_2 < m_3 < m_4 < m_5 \leq M$ and for all $1\leq n_1 < n_2 < n_3 < n_4 \leq N$ such that $m\notin\{m_1,m_2,m_3,m_4,m_5\}$ and $n\notin\{n_1,n_2,n_3,n_4\}$\\[-1mm]

\hspace{5mm}{\bf 3.1.} $\big\{\hat{\alpha}_{n|n_1n_2n_3n_4}, \hat{\beta}_{n|n_1n_2n_3n_4}, \hat{\gamma}_{n|n_1n_2n_3n_4}\big\} \longleftarrow \eqref{eq42}$ when\\[-2mm]

\hspace{12mm} $N$-indices $\{1,2,3,4,5\}\leftarrow\{n,n_1,n_2,n_3,n_4\}$,~ $M$-indices $\{1,2,3,4,5,6\}\leftarrow\{m,m_1,m_2,m_3,m_4,m_5\}$\\[-1mm]

\hspace{5mm}{\bf 3.2.} $\breve{\bf x}_m^{\text{\tiny(n)}} ~\overset{\eqref{eq44}}{\longleftarrow}~ \big\{\hat{\alpha}_{n|n_1n_2n_3n_4}\,,\, \hat{\beta}_{n|n_1n_2n_3n_4}\,,\,\hat{\gamma}_{n|n_1n_2n_3n_4}\big\}$\\[-1mm]

\hspace{5mm}{\bf 3.3.} add the candidate $\breve{\bf x}_m^{\text{\tiny(n)}}$ to $\mathcal{C}_{mn}$.\\[-1mm]

{\bf 4.} Re-estimate $\hat{\bf x}_m^{\text{\tiny(n)}}$ by taking the \emph{median} of all elements in $\mathcal{C}_{mn}$.\\[-2mm]

{\bf Output:} $\big\{\hat{\bf x}^{\text{\tiny(1)}}_m \,\leftrightarrow\,\hat{\bf x}^{\text{\tiny(2)}}_m \,\leftrightarrow\,\cdots \,\leftrightarrow\,\hat{\bf x}^{\text{\tiny(N)}}_m\big\}^M_{m=1}$~:~ Correspondences\\[-2mm]

\hspace{15mm} in which the outliers in $\mathcal{O}$ are re-estimated.\\[-3mm]
}
\end{algorithm}

\subsection{Main algorithm for refinement}

Three applications for the refinement of correspondences are proposed separately in the last three sections. This section studies how to combine these three applications to get an optimal algorithm on refining correspondences. Firstly, we mention that the \emph{Outliers recognition} depends on the \emph{Correspondences refinement} as showing in its step 2.1. In addition, if there are no outliers finding by the \emph{Outlier recognition}, the refinements from the \emph{Correspondences refinement} within and without the two algorithms \emph{Outliers recognition} and \emph{Self estimations} are the same. It is obvious that the refinement from the \emph{Correspondences refinement} will be better if we remove all outliers from its input. Therefore, an iterative process - in which the \emph{Correspondences refinement} will help the \emph{Outliers recognition} on detecting correctly outliers, and vice versa the \emph{Outliers recognition} helps the \emph{Correspondences refinement} to obtain better refinements when some outliers are removed or recovered - is needed. Secondly, to guarantee the very good input for the \emph{Correspondences refinement}, the threshold $\theta$ used on detecting outliers should be small. As a consequence of the small threshold $\theta$, some good observed image points can be recognized as outliers. However if all the remain image points (not the outliers detected by the \emph{Outliers recognition}) are good, the \emph{Self estimation} will give back good values for these good observed image points. Moreover, this algorithm can give good values for the real outliers.

\newpage

From the above two discussions, we propose a combination of the three algorithms \emph{Correspondences refinement}, \emph{Outliers recognition} and \emph{Self estimation} as a final algorithm for refining the observed correspondences. It is the \emph{Main algorithm} and given by Algorithm~\ref{Alg:Main}. There are $\kappa$ iterations and each iteration is an arrangement that the \emph{Correspondences refinement} and then the \emph{Outliers recognition} are implemented via the step 1, then the \emph{Self estimation} is implemented via the step 2 and the \emph{Correspondences refinement} is again finished via the step 3. In each iteration we use the different threshold $\theta_i$ to detect the outliers $\mathcal{O}_i$. These threshold are a descending series, i.e., $\theta_1 > \theta_2 > \cdots > \theta_{\kappa}$.

Logically, the iterative is stopped when the threshold is too small but there are no outlier finding by the step 2. To simplify, we fix the number of iterative $\kappa$ to three and choose the thresholds of 60, 40 and 20 for each iterative in the hope that the final refinement will have an error of about 50 pixels.

\begin{algorithm}[t]
\caption{Main Algorithm}
\label{Alg:Main}
{
{\bf Input:} $\big\{\hat{\bf x}^{\text{\tiny(1)}}_m \,\leftrightarrow\,\hat{\bf x}^{\text{\tiny(2)}}_m \,\leftrightarrow\,\cdots \,\leftrightarrow\,\hat{\bf x}^{\text{\tiny(N)}}_m\big\}^M_{m=1}$~:~ Correspondences\\[-2mm]

\hspace{11mm}~ $\theta_1,\theta_2,\ldots,\theta_{\kappa}$: A sequence of thresholds to recognize outliers\\[-2mm]

{\bf Implementation:}\\[-2mm]

\textbf{\emph{Repeat $\kappa$ times the following steps:}}\\[-3mm]

(Assuming that we are at $i$th time)\\[-1mm]

{\bf 1.} $\mathcal{O}_i ~\overset{\text{Outliers Recognition}}{\longleftarrow}~ \Big\{\big\{\hat{\bf x}^{\text{\tiny(1)}}_m \,\leftrightarrow\,\cdots \,\leftrightarrow\,\hat{\bf x}^{\text{\tiny(N)}}_m\big\}~,~~\theta_i\Big\}$.\\

{\bf 2.} $\big\{\breve{\bf x}^{\text{\tiny(1)}}_m \,\leftrightarrow\,\cdots \,\leftrightarrow\,\breve{\bf x}^{\text{\tiny(N)}}_m\big\}\overset{\text{Self Estimation}}{\longleftarrow}\Big\{\big\{\hat{\bf x}^{\text{\tiny(1)}}_m \,\leftrightarrow\,\cdots \,\leftrightarrow\,\hat{\bf x}^{\text{\tiny(N)}}_m\big\}~,~\mathcal{O}_i\Big\}$.\\

{\bf 3.} $\big\{\bar{\bf x}^{\text{\tiny(1)}}_m \,\leftrightarrow\,\cdots \,\leftrightarrow\,\bar{\bf x}^{\text{\tiny(N)}}_m\big\}\overset{\text{Correspondences Refinement}}{\longleftarrow}\big\{\breve{\bf x}^{\text{\tiny(1)}}_m \,\leftrightarrow\,\cdots \,\leftrightarrow\,\breve{\bf x}^{\text{\tiny(N)}}_m\big\}$.\\

{\bf 4.} $\hat{\bf x}_m^{\text{\tiny(n)}} \longleftarrow \bar{\bf x}_m^{\text{\tiny(n)}}$ ~~for all~~ $m,n$, ~~(Update for the next time)\\

{\bf Output:} $\big\{\bar{\bf x}^{\text{\tiny(1)}}_m \,\leftrightarrow\,\bar{\bf x}^{\text{\tiny(2)}}_m \,\leftrightarrow\,\cdots \,\leftrightarrow\,\bar{\bf x}^{\text{\tiny(N)}}_m\big\}^M_{m=1}$~:~ Refinements.\\[-1mm]
}
\end{algorithm}

\section{Experiments and Evaluations}\label{sec:evaluation}

\subsection{Experiment Settings}
Experiments for studying this work are from our small project on 3D-reconstructing the Buddha statue. The Buddha statue is marked by 200 small black dots. We divide the statue by 21 areas such as the face, the left hand, the right hand, the back, the left leg, the right leg, etc. Two hundreds marked black dots are chosen such that each area has at least ten and at most fifteen marked dots. To reconstruct one area we take around 20 to 30 images. Thus, for each area we have an experiment to study the $N$-view and $M$-point problem where $N$ is the number of images ($20 \leq N \leq 30$) and $M$ is the number of marked dots appearing in this area ($10\leq M \leq 15$).

\begin{figure}[t]
  \centering
  \includegraphics[width=5in]{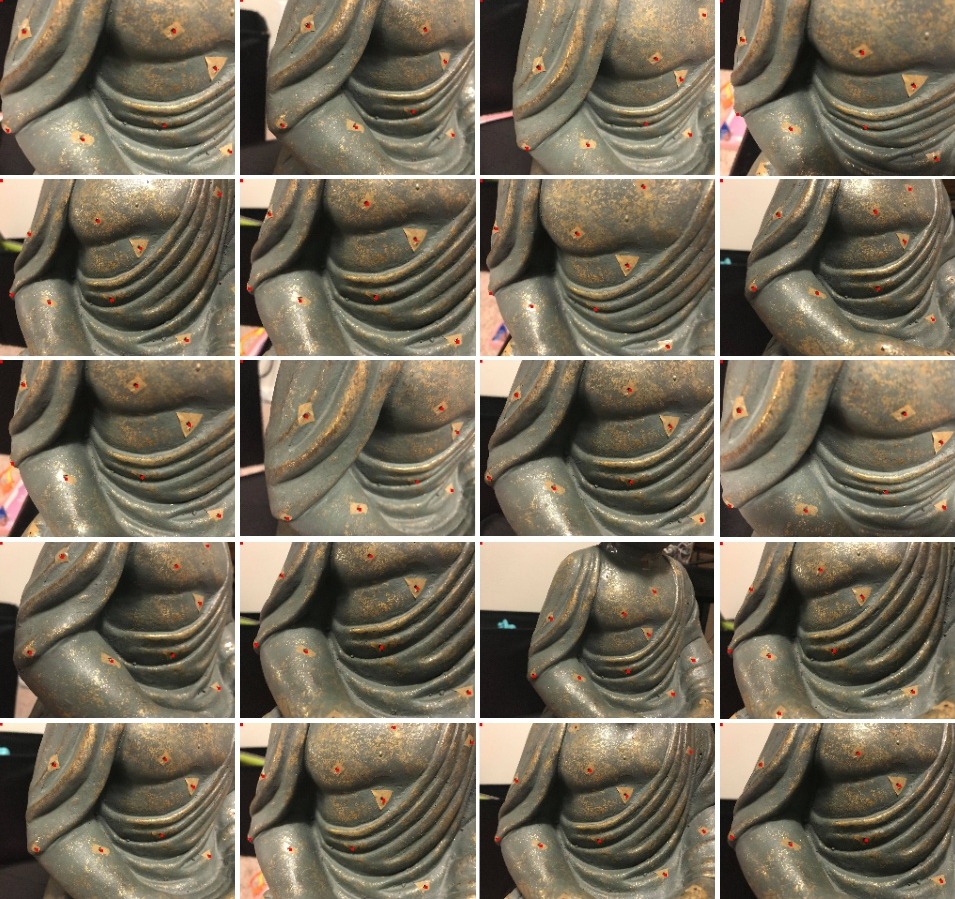}
  \caption{An experiment of the 20-view and 10-point problem for the right front body reconstruction. This is one from the twenty one experiments corresponding to twenty one areas of the Buddha statue. The \emph{red dots} are the ground-truths for the ten-correspondences.}
  \label{Fig:20View10Points}
\end{figure}

To study the correspondences refinement, we manually determine the marked dots in each image. Since the marked black dots are very clear in images, these manual coordinates are very accurate and we use them as the ground-truths for the correspondences. Note that in the experiment of $N$-view and $M$-point, some images do not capture all $M$ marked dots. If the marked dot is missed on the image, its coordinates will be -1. Concretely, let us denote
\begin{figure*}[t]
  \centering
  \includegraphics[width=6in]{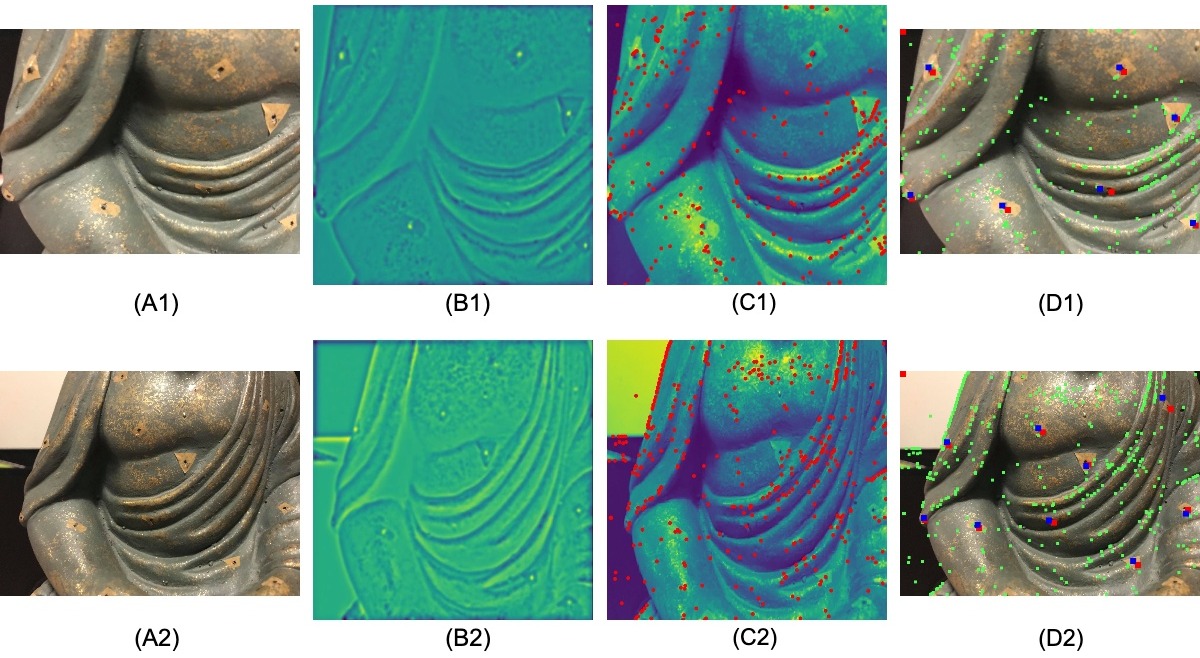}
  \caption{A process on estimating the correspondences based on the ground-truths and SIFT and SURF algorithms. (A): original images, (B): one of the \emph{difference of Gaussian maps}, (C): images with the key-points found by SIFT and SURF, and (D): image-points (\emph{blue squares}) are found from the key-points (\emph{green circles}) closest to the ground truths (\emph{red squares}).}
  \label{Fig:SIFT_SURF}
\end{figure*}
\begin{equation}\label{eq59}
\big\{\overset{\circ}{\bf x}\hspace{-1mm}~_m^{\text{\tiny(1)}} \leftrightarrow\overset{\circ}{\bf x}\hspace{-1mm}~_m^{\text{\tiny(2)}} \leftrightarrow \overset{\circ}{\bf x}\hspace{-1mm}~_m^{\text{\tiny(3)}} \leftrightarrow \cdots \leftrightarrow \overset{\circ}{\bf x}\hspace{-1mm}~_m^{\text{\tiny(N)}}\big\}^M_{m=1}
\end{equation}
be the ground-truths for the correspondences of the $N$-view and $M$-point. The ground-truth $\overset{\circ}{\bf x}\hspace{-1mm}~_m^{\text{\tiny(n)}}$ is $[-1,-1]^T$ if the $m$th marked dot is not captured by the $n$th image. Figure~\ref{Fig:20View10Points} demonstrates an experiment of the 20-view and 10-point reconstruction for the area `\emph{right front body}'. The ground-truths of ten-correspondences in each image are manually highlighted by red dots. The first image at the position (1,1) has only seven motions, the fifteenth image at the position (4,3) has nine motions, and no image captures all ten.

The estimations of the correspondences are found based on the SIFT~\citep{Lowe1999} and SURF~\citep{Bay2006} algorithms. The code for SIFT is based on~\citep{Lowe2004} and the code for SURF is based on~\citep{Evans2009}. A process of estimating the correspondences is explained by the pictures in Figure~\ref{Fig:SIFT_SURF}. Let us explain how to estimate the correspondences on the first image in the experiment given by Figure~\ref{Fig:20View10Points}. The picture (A1) is the original image capturing the right front hand of the Buddha's statue. The seven image points from the ten-correspondences of this experiments are marked by seven black dots. (B1) is one of the `difference of Gaussian' (DOG) map of the image. The key-points are determined based on these DOG maps. All the key-points found by the SIFT and SURF algorithms are red dots in the picture (C1). From lots of key-points, since the ground-truths of the seven image points are known, as highlighted by the `red dots' in picture (D1), we easily choose the closest one and use it as the estimation for each image point. The estimations of the seven image points are highlighted by the `blue dots', and the remain key-points are indicated by the `small green dots'. The Euclidean distances between the estimations and the ground truths are 48.54, 43.43, 32.20, 33.36, 53.60, 129.02 and 28.45 (pixels). A same process on estimating nine image points from the correspondences in the nineteenth image is presented by (A2), (B2), (C2) and (D2). The differences between the estimations and the ground truths are 36.02, 38.35, 76.00, 150.83, 45.81, 31.00, 48.62, 43.65 and 102.86. The average errors of the estimations on the first image and the nineteenth image are \textbf{52.66} pixels and \textbf{63.68} pixels, respectively.

\newpage

Note that without the ground truths, estimating the correspondences is a difficult problem. Since this work is focused on studying the correspondences refinement, not the correspondences estimation, we simplify this step by manually determining the ground-truths for the correspondences. Thanks for this hard work, we have total 21 experiments with 402 images and 4398 image points. For each experiment, the ground truths are really accurate, the estimations are reliable for some real applications on the 3D reconstruction. Given the estimation
\begin{equation}\label{eq60}
\big\{\hat{\bf x}_m^{\text{\tiny(1)}} \leftrightarrow\hat{\bf x}_m^{\text{\tiny(2)}} \leftrightarrow \hat{\bf x}_m^{\text{\tiny(3)}} \leftrightarrow \cdots \leftrightarrow \hat{\bf x}_m^{\text{\tiny(N)}}\big\}^M_{m=1}
\end{equation}
we shall apply the \emph{Main algorithm} (Algorithm~\ref{Alg:Main}) to evaluate how efficient this work is on the correspondences refinement. The number times is three ($\kappa = 3$) and the sequence of thresholds $\{\theta_1, \theta_2, \theta_3\}$ is $\{60, 40, 20\}$. Thus, there are three refinements $\{\bar{\bf x}_m^{\text{\tiny(n)}}\}$ for the estimations $\{\hat{\bf x}_m^{\text{\tiny(n)}}\}$ corresponding to (i) the first iteration with the threshold $\theta_1 = 60$, (ii) the first and second iterations with the thresholds $(\theta_1,\theta_2) = (60, 40)$, and (iii) the output of the \emph{Main algorithm} in which we use all three iterations. We call these three refinements by the `refinements with one iteration', the `refinements with two iterations' and the `refinements with three iterations', respectively. The Euclidean distances $\|\hat{\bf x}_m^{\text{\tiny(n)}} - \overset{\circ}{\bf x}\hspace{-1mm}~_m^{\text{\tiny(n)}}\|_2$ and $\|\bar{\bf x}_m^{\text{\tiny(n)}} - \overset{\circ}{\bf x}\hspace{-1mm}~_m^{\text{\tiny(n)}}\|_2$ are called the \emph{point-errors} of the estimated and refined image points $\hat{\bf x}_m^{\text{\tiny(n)}}$ and $\bar{\bf x}_m^{\text{\tiny(n)}}$. The mean point-errors
\begin{equation}\label{eq61}
\frac{1}{M}\sum^M_{m=1}\big\|\hat{\bf x}_m^{\text{\tiny(n)}} - \overset{\circ}{\bf x}\hspace{-1mm}~_m^{\text{\tiny(n)}}\big\|_2 \quad \text{and} \quad \frac{1}{M}\sum^M_{m=1}\big\|\bar{\bf x}_m^{\text{\tiny(n)}} - \overset{\circ}{\bf x}\hspace{-1mm}~_m^{\text{\tiny(n)}}\big\|_2
\end{equation}
are called the \emph{image-errors} of all the estimated and refined image points on the $n$th image. We call these errors shortly by the \emph{$n$th image-errors}. Finally, the mean image-errors
\begin{equation}\label{eq62}
\frac{1}{NM}\sum_{n,m}\big\|\hat{\bf x}_m^{\text{\tiny(n)}} - \overset{\circ}{\bf x}\hspace{-1mm}~_m^{\text{\tiny(n)}}\big\|_2 \quad \text{and} \quad \frac{1}{NM}\sum_{n,m}\big\|\bar{\bf x}_m^{\text{\tiny(n)}} - \overset{\circ}{\bf x}\hspace{-1mm}~_m^{\text{\tiny(n)}}\big\|_2
\end{equation}
are called the \emph{correspondence-errors} of the estimated correspondences $\{\hat{\bf x}_m^{\text{\tiny(1)}}\leftrightarrow\hat{\bf x}_m^{\text{\tiny(2)}}\leftrightarrow\cdots\leftrightarrow\hat{\bf x}_m^{\text{\tiny(N)}}\}$ and the refined correspondences $\{\bar{\bf x}_m^{\text{\tiny(1)}}\leftrightarrow\bar{\bf x}_m^{\text{\tiny(2)}}\leftrightarrow\cdots\leftrightarrow\bar{\bf x}_m^{\text{\tiny(N)}}\}$. These errors together will give us a good evaluation of the correspondences refinements. Lastly, in the \emph{Main algorithm} we use, the number of iterations in the \emph{Correspondences refinement} is 10.

\subsection{Evaluations}

Figure~\ref{Fig:OneExperiment} gives an evaluation of the $21$-view and $12$-point experiment for the area `front and bottom body'. The upper sub-figure presents the image-errors of the estimations and the three refinements for twenty one images. This result shows that for all twenty one images, the refinements with two iterations are always better than one with one iteration, and the final refinements corresponding with three iterations are the best. Except for the seventeenth image, all the refinements are better than the estimations. Figure~\ref{Fig:Simulated_One} shows the ground-truths, the estimations and the three refinements of the ten from twelve image points on the twenty first image and the thirteenth image. Based on the third refinements, the twenty first image is the best and the thirteenth image is the worst. The image-errors on the best are 55.69 (the estimations), 44.24 (the first refinements), 35.07 (the second refinements) and 33.93 (the third refinements). The image-errors on the worst are 116.79 (the estimations), 106.24 (the first refinements), 81.59 (the second refinements) and 72.67 (the third refinements). Since the differences between the estimations and the ground-truths on the twenty first image are small, all three its refinements are good. The point-errors of the fifth and ninth image points on the thirteenth image are very large, so the refinements for this image are not good. Interestingly, the differences between the estimations and the ground-truths on the seventh image are large, but its three refinements are very good. This is an good example for us to see the role of the \emph{Self estimation} algorithm on refining the correspondences.

\begin{figure}[t]
  \centering
  \includegraphics[width=4.5in]{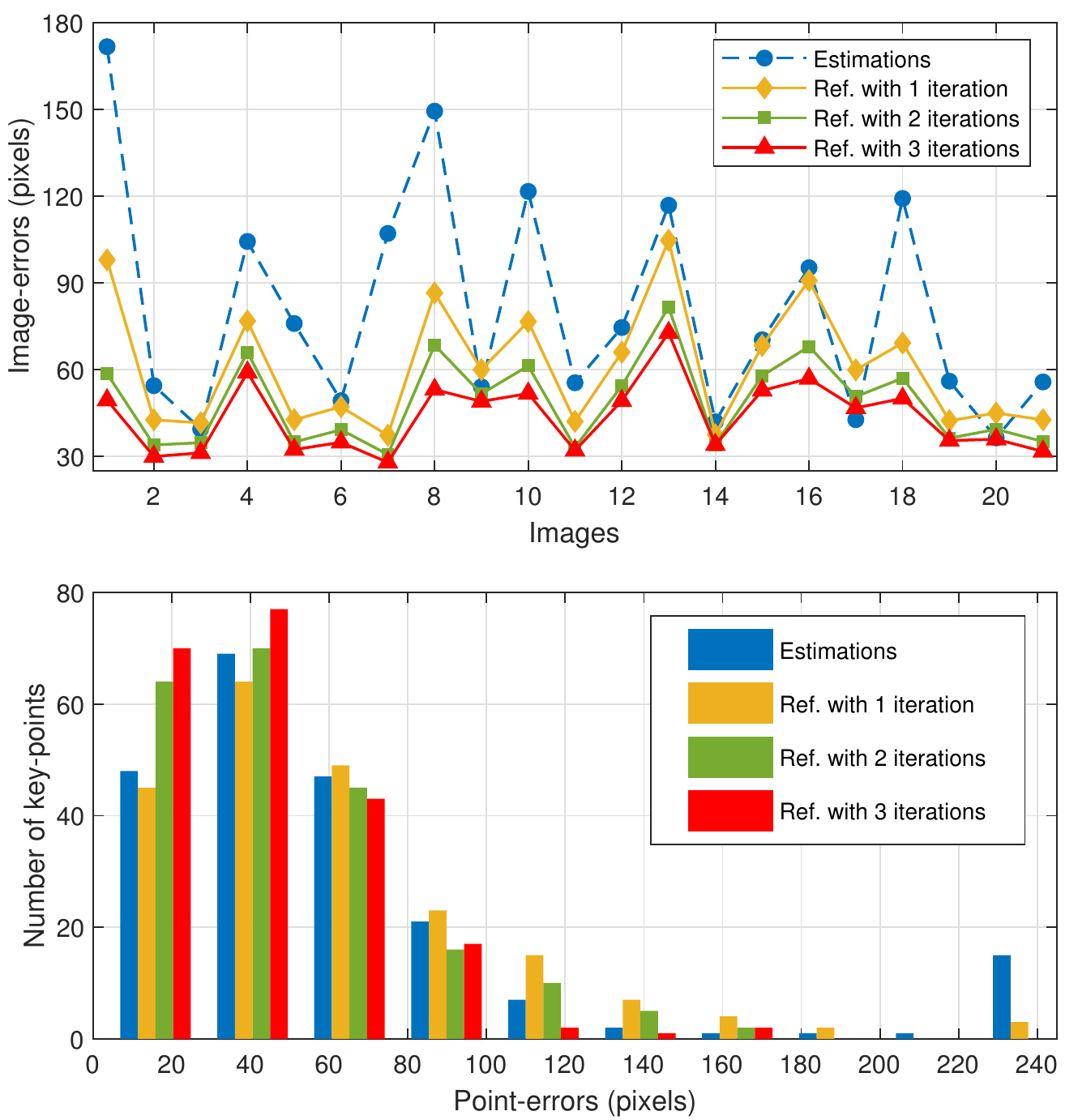}
  \caption{An evaluation of the 21-view and 12-motion experiment. \emph{Upper:} The image-errors of twenty one images. \emph{Lower:} The histogram of all point-errors.}
  \label{Fig:OneExperiment}
\end{figure}

\begin{figure}[t]
  \centering
  \includegraphics[width=6.25in]{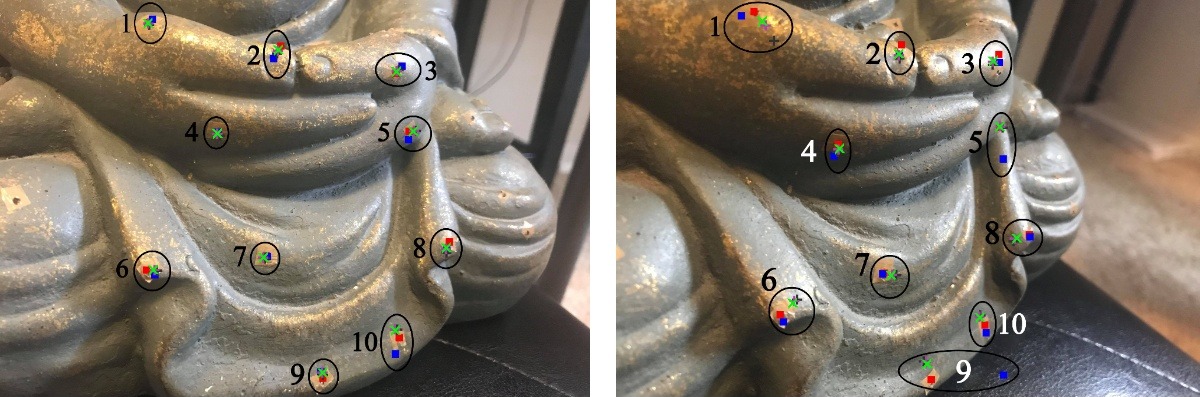}
  \caption{The ground-truths and the estimated and refined image points on the twenty first image (upper) and the thirteenth image (lower). The ground-truths are the \emph{red squares}, the estimations are the \emph{blue squares}, the refinements with one iteration are the 'black plus symbols', the refinements with two iterations are the `pink plus symbols', and the refinements with three iterations are the `green cross symbols'. The ground-truth and its estimations, refinements are grouped by the `black ellipses' and indicated by the corresponded number.}
  \label{Fig:Simulated_One}
\end{figure}

The lower sub-figure in Figure~\ref{Fig:OneExperiment} presents the histogram of all point-errors on the discussed 21-view and 12-motion experiments. If an outlier is defined by its point-error greater than 120 and a good image point is defined by its point-error less than 50, this histogram tells us that there are around 20 outliers and 100 good image points in the estimations, and around 4 outliers and 140 good image points in the refinements with three iterations. The smaller number of outliers in the refinements comparing with the estimations validates the role of the \emph{Outlier recognition}, and the larger number of good image points in the refinements comparing with the estimations again validates the efficient of the \emph{Self estimation}. Finally, the correspondence-errors of this experiment are 80.95 (the estimations), 63.76 (the refinements with one iteration), 48.63 (the refinements with two iterations) and 43.34 (the refinements with three iterations).

\begin{figure}[h!]
  \centering
  \includegraphics[width=4.5in]{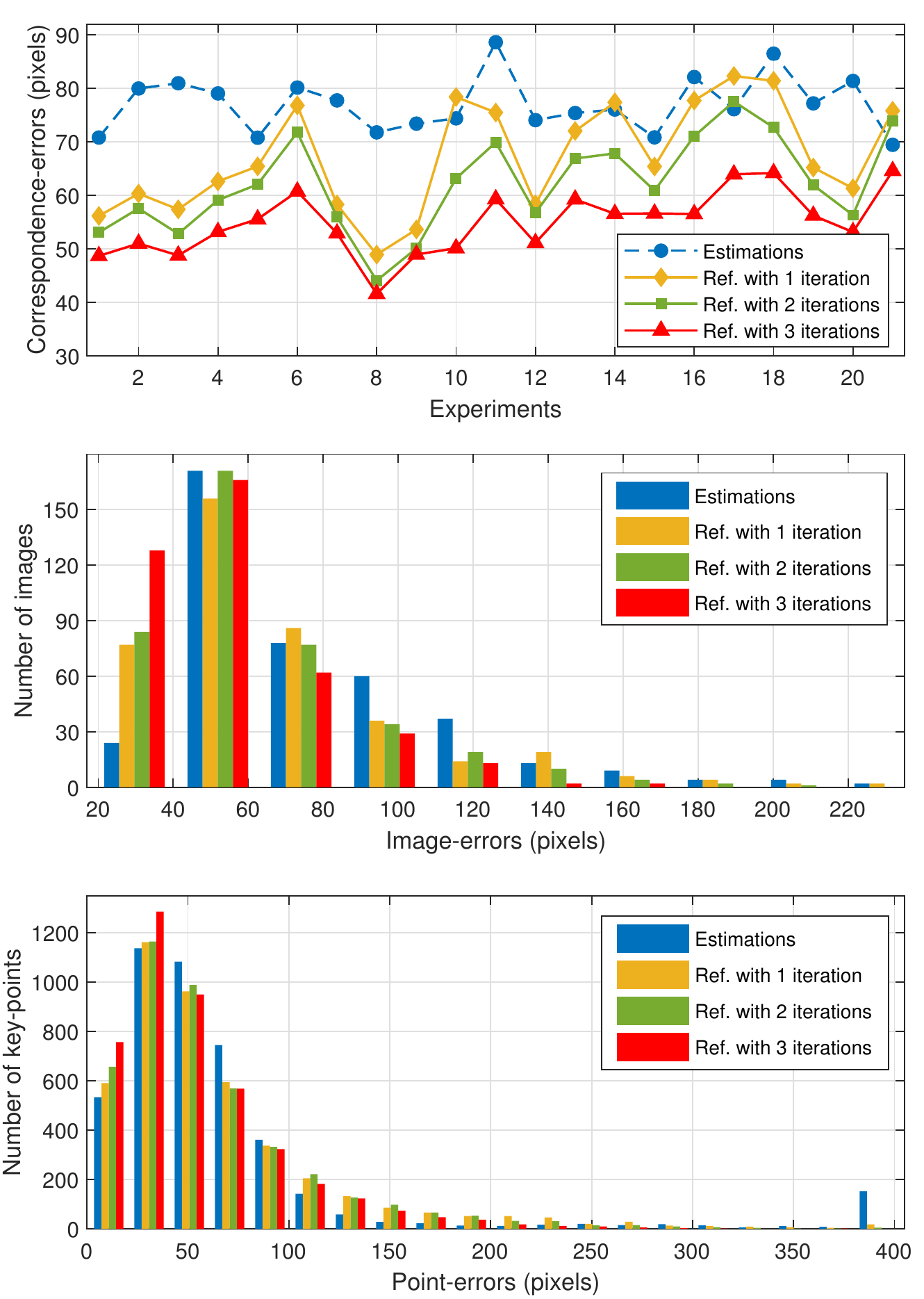}
  \caption{Evaluations of twenty one experiments from the 3D reconstruction of Buddha statue. These evaluations are based on the correspondence-error for the correspondences, the image-error for the images and the point-error for the key-points.}
  \label{Fig:AllExperiments}
\end{figure}

Figure~\ref{Fig:AllExperiments} presents the correspondence-errors, image-errors and point-errors for all 21 experiments. The first sub-figure presents all the correspondence-errors given by Eq.~\eqref{eq62} for all 21 experiments. The result shows that the correspondences refinement proposed by this paper works very well for the eighth experiment but does not work well for the last experiment. It is important to note that the refinements with two iterations are always better than the refinements with one iteration, and the refinements with three iterations are always the best. The second sub-figure simulates the histogram of 402 image-errors given by Eq.~\eqref{eq61} for 402 image and the last sub-figure simulates the histogram of 4398 point-errors. The average of the point-errors corresponding to the estimations is 76.93, and ones corresponding to the refinements with one, two and three iterations are 67.29, 62.29 and 55.15, respectively. Note that the averages of the point-errors and the image-errors are the same. When we replace the average by the median, these values from the histogram of the point-errors are 50.15 (the estimations) and (49.32, 47.59, 43.20) (the refinements with one, two and three iterations). These numbers say that 50\% estimated image points have errors of less than 50 pixels, and 50\% final refinements have errors of less than 43 pixels. The percentage of final refinements with an error of less than 50 pixels is 57\%. Let us detect the outliers by its point-error greater than 200, then there are 274 outliers from the estimations, 210 outliers from the refinements with one iteration, 121 outliers from the refinements with two iterations, and 51 outliers from the refinements with three iterations. The medians from the histogram of the image points are 63.95 (the estimations) and (56.05, 51.99, 48.16) (the refinements with one, two and three iterations). If we define the good image by its image-error less than 50, then the refinements with three iterations has more than 50\% good images and the estimations has only 23\% good image.

Overall, evaluating based on the ground-truths, an improvement of the \emph{Main algorithm}, comparing with the estimations and the other two refinements corresponding with one and two iterations, is consistent.

\newpage

\section{Conclusion}\label{sec:conclusion}

This paper concentrates on improving the correspondences accuracy, the most important problem in the 3D reconstruction based on multiple images. Assuming that correspondences are estimated from some existing methods, this paper proposes an independent algorithm which is based on two rank deficient matrices to refine them. The first $M\times 9$ rank deficient matrix is constructed using correspondences from two images and $M$ key-points. When the number of key-points is at least nine ($M\geq 9$), the rank deficient property of this matrix confirms its minimum singular value should be zero. Hence, if this value is positive, we know that there is noise. Of course if this value is too large, we really believe that there are bad correspondences and have some outliers in them. The second $N\times 5$ rank deficient matrix is constructed using correspondences from $N$ images and six key-points. Similar to the first one, when there are at least five images ($N\geq 5$), this rank deficient property helps us to recognize noise and outliers in the correspondences. These two rank deficient matrices are found and proven by Theorem~\ref{Thm:RanktwoView} and Theorem~\ref{Thm:RankfiveView}. \\

\noindent The proposed correspondences refinement algorithm, Algorithm~\ref{Alg:Main}, is validated by a number of practical experiments. To study carefully the correspondences refinement and other future researches such as the world points estimation (the triangulation), the building point clouds and surface (the bundle adjustment), etc., we generate a lot of `\emph{reliable}' experiments on the 3D-reconstruction of the Buddha statue project. From this project, we have 402 images capturing 200 special world points marked on the Buddha statue. Some of these 200 world points captured by the image are defined as the key-points on this image. Since they are very clear by looking, we manually determine the ground-truths for all key-points on all images. On another hand, the Euclidean distances between these 200 world points are measured to determine their relative positions. Thus, we emphasize that there are good experiments in which the ground truths for the correspondences, the world points and the projection matrices are highly accurate, for us to study the correspondences estimation, the correspondences refinement, the world points estimations, the building point clouds and surface, and more. \\

\noindent The results working on the above reliable experiments demonstrate that most of the correspondences will be improved by the proposed algorithm on the point of recognizing outliers, recovering missing key-points and reduce point-errors. The refined correspondences we obtain are sufficiently good for us to continue our researches on the world points estimation and the building point clouds.\\

\noindent One limitation of this paper is that it has not exploited enough Theorem~\ref{Thm:RankfiveView} to refine correspondences. We believe that if Theorem~\ref{Thm:RankfiveView} is used better then \emph{Main algorithm} will be better. Our future work will focus on how to exploit Theorem~\ref{Thm:RankfiveView} to improve our current results.

\bibliographystyle{plainnat}
\bibliography{refs}

\end{document}